%% file: acfgp-hetalpha-FULL.tex
\documentclass[a4paper,10pt]{article}
\usepackage[utf8]{inputenc}
\usepackage{fullpage}
\usepackage{amsmath,amssymb,amsthm}
\usepackage[noend,linesnumbered,noline,nofillcomment]{algorithm2e}
\usepackage{tikz}
\usetikzlibrary{positioning}
\usepackage{hyperref}

\newcommand{\blf}{\mathbf{b}}

\renewcommand{\S}{{\mathcal{S}}}
\newcommand{\notS}{{\overline{S}}}
\newcommand{\nS}{{\overline{S}}}
\newcommand{\notT}{{\overline{T}}}
\newcommand{\T}{{\mathcal{T}}}

\newcommand{\notN}{{\overline{N}}}
\newcommand{\defi}{{\mathsf{def}}}
\newcommand{\rank}{{\mathsf{rank}}}
\newcommand{\ox}{\overline{x}}

\newcommand{\st}{{\mathbf{s}}}

\newtheorem{theorem}{Theorem}
\newtheorem{lemma}{Lemma}

\theoremstyle{definition}
\newtheorem{definition}{Definition}

\theoremstyle{remark}
\newtheorem{remark}{Remark}

\title{{\bf Discrete Preference Games in Heterogeneous Social Networks}\\Subverted Majorities and the Swing Player}
\author{Vincenzo Auletta, Ioannis Caragiannis, Diodato Ferraioli,\\ Clemente Galdi, and Giuseppe Persiano}
\date{}

\begin{document}

\maketitle

\begin{abstract}
We study {\em discrete preference} games in {\em heterogeneous} social networks.
These games model the interplay between a player's {\em private belief} 
and his/her {\em publicly stated opinion} (which could be different from the player's belief) 
as a strategic game in which the players' strategies are the opinions and 
the cost of an opinion in a state is a convex combination through a parameter $\alpha\in[0,1]$ of two factors:
the disagreement between the player's opinion and his/her internal belief 
and the number of neighbors whose opinions differ from the one of the player.
The parameter $\alpha$ models how stubborn a player is: players with large $\alpha$ change 
their opinion only if many neighbors disagree with his/her belief. 
We consider social networks that are {\em heterogeneous} in the sense that the parameter 
$\alpha$ can vary from player to player.

We ask whether decisions made by a social network are robust to internal pressure and  
investigate the phenomenon by which, because of {\em local} strategic decisions 
at the level of the players, 
the {\em global} majority can be {\em subverted}.
More precisely, we ask if it is possible that the belief shared by the majority of the 
players does not coincide with the opinion that is publicly announced by the majority 
of the players in an equilibrium state. 

Our main result is a characterization of the social networks that admit an initial belief assignment 
for which there exists a sequence of best response moves that reach an equilibrium
in which the initial majority is subverted.  
Our characterization is effective in the sense that can be tested efficiently and 
the initial belief assignment that can be subverted can be computed in time polynomial 
in the number of players.
Our result is actually stronger as 
we show that in each initial belief assignment that can be subverted, 
subversion is actually obtained in a very strong way:
there exists {\em one} player, the {\em swing} player, 
that changes his/her opinion to improve his/her utility and,
as a result of this best response move, every subsequent sequence of best response 
moves of the other players leads to an equilibrium in which majority is subverted.
In other words, it only takes one move of the swing player to lead the social network to 
a point of no return in which any rational move from any player leads to a subverted majority.
\end{abstract}

\section{Introduction}
Social networks play a crucial role in the way individuals form their opinions every day.
Just to name few examples, 
a recruiter interviewing a job-seeker can be influenced by the opinion that coworkers expressed on 
Facebook or LinkedIn;
a tourist looking for a restaurant can be influenced on the opinion read on Yelp or TripAdvisor;
a traveler looking for a hotel is influenced by the opinion given by his friends on Booking.com.

Several models have been introduced in order to model how agent react to this influence
and how the opinions are formed.
A first natural model is given by majority dynamics: each agent has an initial
preference and at each time step a subset of players updates their
opinion to the one expressed by the majority of their neighbors. 
Another classical model has been proposed by DeGroot \cite{D74}, and later refined by Friedkin and Johnsen \cite{FJ90}.
The model assumes that each agent has a {\em private belief},
but the public {\em opinion} she eventually expresses can be different from her belief.
In particular, the opinion comes out from a repeated averaging between the
belief and the opinions of individuals with whom she has social relations.
A similar model, that is more suitable for the case in which beliefs and opinions are discrete,
assumes that each agent behaves strategically and aims to pick the most beneficial (or less costly)
opinion for her, where the benefit (or cost) depends both on her internal belief
and on the opinions of individuals with whom she has social relations.
This model has been recently adopted by Ferraioli {\em et al.}~\cite{fgvSAGT12} and by Chierichetti {\em et al.}~\cite{ckoEC13},
that considered the case of binary beliefs or opinions\footnote{Chierichetti {\em et al.}~\cite{ckoEC13}
also consider the case of multiple beliefs/opinions. Still, in this case multiple metrics can be considered 
for evaluating the distance among opinions, 
and different metrics embeds a very different behavior of people in social networks.}.

\paragraph*{Our setting.}
In this paper we consider binary {\em discrete preference games} with $n$ players.
A discrete preference game consists of a $n$-vertex undirected graph $G$ (the social network), 
{\em stubbornness} levels $\alpha_1,\ldots,\alpha_n\in (0,1)$ and {\em beliefs } 
$\blf(1),\ldots,\blf(n)\in\{0,1\}$.
The players are identified with the vertices of $G$ and each player $i$ has 
a {\em stubbornness} level $\alpha_i\in (0,1)$ and a {\em belief} $\blf(i)\in\{0,1\}$.
Player $i$ can choose to play {\em opinion } $\st(i)\in\{0,1\}$ and 
the cost $c_i(\st)$ of agent $i$ in state 
$\st=(\st(1), \ldots, \st(n))\in\{0,1\}^n$ is defined as 
\begin{equation}
 \label{eq:cost}
c_i(\st) = \alpha_i \cdot \left|\st(i) - \blf(i)\right|+\left(1-\alpha_i\right) \cdot \sum_{j \in N(i)}{\left|\st(i) - \st(j)\right|},
\end{equation}
where $N(i)$ is the set of neighbors of vertex $i$ in $G$ (i.e., friends in the social network). 
Note that the cost is the convex combination through $\alpha_i$ of two components that depend
on whether the opinion coincides with the belief and on the strategies of the neighbors, respectively.
The stubbornness level $\alpha_i$ measures the weight given by player $i$ to her own belief. 
Roughly speaking, high values of $\alpha_i$ are associated with players that need a lot of convincing from 
their friends to change their mind and adopt an opinion in contrast with their beliefs.

We consider the game starting in the {\em truthful} state in which $\st(i)=\blf(i)$ for all $i$ and 
then evolving through a set of sequential best response moves until an equilibrium is reached.
We define an {\em equilibrium} state to be a state $\st=(\st(1),\ldots,\st(n))$ for which
there is no player $i$ whose best response is to adopt strategy $1-\st(i)$. 
More precisely, $\st$ is an {\em equilibrium} if for all $i$
$$c_i(\st)\leq c_i((\st_{-i},1-\st(i))),$$
where we have used the standard game theoretic notation by which 
$(\st_{-i},a)$ denotes the vector 
$$(\st_{-i},a)=(\st(1),\ldots,\st(i-1),a,\st(i+1),\ldots,\st(n)).$$

Obviously, the strategic moves of the player affect the way the network works and evolves.
Hence, in order to get an insight about how we can improve the performance of the social network,
we need to understand at which extent the behavior of network's members may affect the system at large.
Is it possible that the {\em local} behavior of the players affect the {\em global} behavior of the network?
Does the social pressure felt by individual members of a social networks have any effects on the entire network?

In this paper we 
take the behavior of a network to be the majority of the opinions expressed by its members
once the network has reached an equilibrium state.
We ask whether it is possible that the majority of the opinions differ from the majority of the belief.
In other words, is it possible that majority can be subverted by social pressure?

\paragraph*{Our contribution.}
We say that a pair $(G,(\alpha_1,\ldots,\alpha_n))$ consisting of a graph $G$ with $n$ vertices 
(we assume $n$ odd so that majority is well defined)  
and of the sequence of stubbornness levels is {\em subvertable} if 
there exists a belief assignment $(\blf(1),\ldots,\blf(n))$ with a majority of $0$ 
and a sequence of best response moves that goes from the truthful state to an equilibrium state with 
a majority of $1$. We call such a belief assignment {\em subvertable}.

Our main contribution is a characterization of the subvertable pairs.  Roughly speaking, 
our characterization says that a pair is subvertable unless all players are {\em stubborn}. 
More precisely, 
consider vertex $x$ with $\blf_x=0$ and $d_0$ neighbors with opinion $0$ and $d_1$ neighbors with opinion $1$.  
Then the cost of opinion $0$ for $x$ is $(1-\alpha_x) d_1$ whereas
the cost of opinion $1$ is $\alpha_x+(1-\alpha_x) d_0$. Therefore $x$ has an incentive to 
declare opinion $1$ iff $d_1-d_0\geq a_x+1$, where $a_x=\lfloor\frac{\alpha_x}{1-\alpha_x}\rfloor$. 
Clearly, this cannot happen if the degree $d(x)=d_1+d_0$ of $x$ satisfies $d(x)\leq a_x$.
Suppose now the majority (that is at least $(n+1)/2$ vertices) have belief $0$ and thus $d_1\leq (n-1)/2$.  
If $d(x)\geq n-a_x-1$ then 
$$d_1-d_0 =  2d_1-d(x) \leq n-1 -d(x) \leq  a_x$$
and thus vertex $x$ does not have an incentive to declare opinion $1$. 
The same reasoning applies for vertices $x$ with $\blf_x=1$ in case majority is $1$. 
We have thus the following definition.
\begin{definition}[Stubborn vertex]
A vertex $x$ with degree $d(x)$ and stubbornness $\alpha_x$ is {\em stubborn} 
if 
$$a_x:=\left\lfloor{\frac{\alpha_x}{1-\alpha_x}}\right\rfloor \geq \min\left\{d(x), n-d(x)-1\right\}.$$
\end{definition}
Clearly, if all vertices are stubborn then majority cannot be subverted as no vertex has an incentive 
to play an opinion be different from the belief. 
The main result of this paper shows that 
\begin{center}
{\em if there exists at least one non-stubborn vertex then there exists a subvertable belief assignment.}
\end{center}
A possible interpretation of our result is that social networks are extremely vulnerable to social pressure
since there always exists a subvertable majority unless all vertices are stubborn and never change their
mind (in which case we do not have much of a social network).
This is particularly negative as an external adversary might be able to orchestrate the sequence of best 
response moves so to reach the state in which majority is subverted. 
In principle, though, this could be very difficult since there could be different sequences of best response moves that lead to different equilibria with different 
majorities and the adversary has to be very careful in scheduling the best response moves.
Our characterization instead proves that there is always one single {\em swing} player 
whose best response in the truthful state is to change her opinion and this leads the social
network to a state in which {\em any} sequence of best response moves leads to an equilibrium in which 
majority has been subverted. In other words, the adversary that wants to subvert the majority only has
to influence the swing player and then the system will evolve without any further intervention towards an
equilibrium in which majority is subverted.
More precisely, 
\begin{definition}
A vertex $u$ is said to be a \emph{swing} vertex for subvertable belief assignment 
$\blf$ with $\frac{n+1}{2}$ vertices with belief $0$ if 
\begin{enumerate}
\item $\blf(u)=0$;
\item $c_u(\blf)>c_u(\blf')$, where $\blf'=(\blf_{-u},1))$. 

    That is, in the truthful state, $u$'s best response is to play opinion $1$.

\item For every $x$ with $\blf(x)=1$, 
it holds that $c_x(\blf')\leq c_x(0,\blf'_{-x})$.

That is, after $u$'s best response no vertex with belief $1$ has an incentive to change her opinion.
\end{enumerate}
\end{definition}
Note that definition above does not imply that the majority at equilibrium consists of only $\frac{n+1}{2}$ vertices
with belief $1$ (the initial $\frac{n-1}{2}$ plus the swing vertex).
It may be indeed the case that other vertices with belief $0$ have an incentive to change their opinion after 
the swing vertex's best response move.
Still, the definition of swing vertex assures that, after her best response, the number of vertices with opinion $1$ 
is a majority and the size of this majority does not decrease.

Our main result can be improved as follows
\begin{center}
{\em if there exists at least one non-stubborn vertex then there exists a subvertable belief assignment with a swing vertex.}
\end{center}

It is natural to ask whether the characterization can be strengthened to take into account
strong majorities (that is, majorities of size at least $(1+\delta)\frac{n+1}{2}$ for some $0<\delta<1$).
That is, to characterize the pairs
(consisting of a social network and stubbornness levels) that admit at least a subvertable strong majority.  
We prove that no such characterization can be given by showing that
there exists $\delta_{\text{max}}\approx 0,85$ such that for all $0<\delta<\delta_{\text{max}}$ 
it is NP-hard to decide whether a given $G$ and given stubbornness levels $\alpha_1,\ldots,\alpha_n$ 
admit a subvertable majority of size at least $(1+\delta)\frac{n+1}{2}$.

\paragraph*{Previous work.}
Our work is strictly related with a line of work in social sciences that
aims to understand how opinions are formed and expressed in a social context.
A classical simple model in this context has been proposed by Friedkin and
Johnsen \cite{FJ90} (see also \cite{D74}). Its main assumption is that each individual has a
private initial belief and that the opinion she eventually expresses is the result
of a repeated averaging between her initial belief and the opinions expressed by
other individuals with whom she has social relations. 
The recent work of Bindel et al.~\cite{BKO11} assumes that initial beliefs and opinions belong to [0, 1] and interprets
the repeated averaging process as a best-response play in a naturally defined
game that leads to a unique equilibrium.

Discrete belief and opinions have been first considered in ~\cite{fgvSAGT12} that studied 
rate of convergence of the game under different dynamics  
and in \cite{ckoEC13} that were mainly interested in the price of stability and price of anarchy of the games.
In a previous paper~\cite{acfgpWINE15}, the authors have studied subvertable majorities
for the {\em majority} dynamics in which players adopt the majority of the opinions expressed by 
the neighbors and uses their private belief only as a tie breaker.
This dynamics corresponds to the special case of {\em homogeneous } networks in which for all $x$, 
$\alpha_x=\alpha$ for some $\alpha<1/2$.

\paragraph*{Notation.}
For subsets $A,B\subseteq V$ of the vertices of $G$ we denote by 
$W(A,B)$ the number of edges with one endpoint in $A$ and the other in $B$.
If $A=\{x\}$ is a singleton, we will simply write $W(x,B)$; similarly for $B$.
Thus, for vertices $x,y$, $W(x,y)=1$ if and only if $x$ and $y$ are adjacent.

\section{Definitions and Technical Overview}
In this section we introduce the concepts of a bisection and of a good bisection and give an 
overview of the proof.

\paragraph*{Good bisections yield subvertable belief assignments.}
A {\em bisection} $\S=(S,\notS)$ of a graph $G$ with an odd number $n$ of vertices is 
a partition of the vertices of $G$ into two sets $S$ and $\notS$ of cardinality $\frac{n+1}{2}$ and 
$\frac{n-1}{2}$, respectively.
We define the {\em deficiency} $\defi_\S(x)$ of a vertex $x$ with respect to bisection $\S=(S,\notS)$ 
as follows:
$$\defi_\S(x)=\begin{cases}
	W(x,S)-W(x,\notS), & \text{if } x\in S;\\
	W(x,\notS)-W(x,S), & \text{if } x\in \notS.\\
	\end{cases}
$$
We say that a bisection $\S=(S,\notS)$ is \emph{good} if
\begin{enumerate}
\item for every $x \in S$, $\defi_\S(x) \geq -a_x$;
\item there is $u \in S$ with $\defi_\S(u) \geq a_u + 1$.
\end{enumerate}
Vertices $u\in S$ with  $\defi_\S(u) \geq a_u + 1$ 
are called the \emph{good vertices} of $\S$ and
vertices $y\in S$ with $\defi_\S(y)<-a_y$ are called the {\em obstructions} of $\S$\footnote{%
We remind the reader that, for vertex $x$, we set $a_x=\left\lfloor\frac{\alpha_x}{1-\alpha_x}\right\rfloor$.
}.
Next lemma proves that if $G$ has a good bisection then one can easily construct 
a subvertable belief assignment for $G$.
\begin{lemma}
\label{lem:good_bisection}
Let $\S=(S,\notS)$ be a good bisection for graph $G$ and let $u$ be one of its good vertices.
Then $G$ admits a subvertable belief assignment $\blf$ such that $u$ is a swing vertex for $\blf$.
\end{lemma}
\begin{proof}
Consider the belief assignment $\blf$ such that $\blf(x) = 1$ for every $x \in S\setminus\{u\}$
and $\blf(x) = 0$ for every $x\in\notS\cup\{u\}$.
Thus, in the truthful profile $\blf$ there is a majority of vertices with opinion $0$.
 
Now, consider vertex $u$. 
Since $\S$ is good and $u$ is a good vertex for $\S$,
$$\defi_\S(u) = W(u,S) - W(u,\notS) \geq a_u + 1 > \frac{\alpha_u}{1-\alpha_u}.$$
The cost of $u$ in the truthful state $\blf$ is $c_u(\blf)=(1-\alpha_u)W(u,S)$; 
if $u$ plays opinion $1$ instead the cost is $c_u((\blf_{-u},1))=\alpha_u+(1-\alpha_u)W(u,\notS)$.
It follows that $c_u(\blf)-c_u((\blf_{-u},1))=(1-\alpha_u)\defi_\S(u)-\alpha_u>0$.
Then it is a best-response for $u$ to adopt opinion $1$.

Let $\blf' = (1,\blf_{-u})$, i.e., the profile reached after the best-response of $u$.
Note that in $\blf'$ there is a majority of vertices with opinion $1$.
We prove that no vertex $x$ with opinion $1$ (that is, no vertex $x \in S$)
has an incentive to change her opinion, from which we can conclude that
$\blf$ is a subvertable belief assignment and $u$ is a swing vertex for $\blf$.
 
This is obvious for $x = u$. 
Since $\S$ is good, then, 
for every $x\in S$, 
$$\defi_\S(x)=W(x,S)-W(x,\notS)\geq -a_x\geq -\frac{\alpha_x}{1-\alpha_x}.$$
As for $x \neq u$, the cost of $x$
in state $\blf'=(\blf_{-u},1)$ is $c_x(\blf')=(1-\alpha_x)W(x,\notS)$
whereas the cost of $x$ in $\blf''=(\blf'_{-x},0)$ is $c_x(\blf'')=\alpha_x + (1-\alpha_x)W(x,S)$.
It follows that 
$c_x(\blf'')-c_x(\blf')=(1-\alpha_x)\defi_\S(x)+\alpha_x\geq 0$
and, thus, $x$ has no incentive to adopt opinion $0$.
\end{proof}

\paragraph*{Minimal bisections.}
The technical core of our proof is the construction of a good bisection starting from a bisection 
$\S$ of minimal potential $\Phi$. 
We define the {\em potential} $\Phi$ of a bisection $(S, \notS)$ as
$$\Phi(S,\notS) = W(S,\notS) + \frac{1}{2}\left(\sum_{x \in S} a_x - \sum_{y \in \notS} a_y\right).$$
We say that a bisection $\S$ has {\em $k$-minimal} potential if $\S$ minimizes the potential among all the
bisections that can be obtained from $\S$ by swapping at most $k$ vertices between $S$ and $\notS$. 
That is, $\S$ has {\em $k$-minimal} potential if, for all $A \subseteq S$ and for all $B \subseteq\notS$, 
with $1\leq |A|=|B|\leq k$,
$$\Phi(S,\notS)\leq \Phi(S \setminus A \cup B,\notS \setminus B \cup A).$$
We will simply write that $\S$ has minimal potential whenever $\S$ has $1$-minimal potential.

Next lemmas prove some useful properties of minimal bisections.
\begin{lemma}
\label{lem:minimality}
Let $\S=(S,\notS)$  be a bisection of minimal potential. Then 
for all $x\in S$ and $y\in\notS$, 
$$\defi_\S(x)+\defi_\S(y)+2 W(x,y)\geq a_x - a_y.$$
\end{lemma}
\begin{proof}
Set $A=S\setminus\{x\}$, $B=\nS\setminus\{y\}$, 
$T=A\cup\{y\} \text{ and } \overline{T}=B\cup\{x\}$.
Note that $$\Phi(T,\overline{T}) = W(A,B)+W(x,A)+W(y,B)+W(x,y)+\frac{1}{2}\left(\sum_{u\in A} a_u - \sum_{v \in B} a_v + a_y - a_x\right)$$ and
$$\Phi(S,\nS) = W(A,B)+W(x,B)+W(y,A)+W(x,y)+\frac{1}{2}\left(\sum_{u\in A} a_u - \sum_{v \in B} a_v + a_x - a_y\right).$$
Since $\S$ has minimal potential we have
\begin{align*}
0 &\leq \Phi(T,\overline{T}) - \Phi(S,\nS) = W(x,A)+W(y,B) - W(x,B)-W(y,A) + a_y - a_x\\
  &=   W(x,S)-W(x,\nS)+W(y,\nS)-W(y,S)+2W(x,y) + a_y - a_x\\
  &=   \defi_\S(x)+\defi_\S(y)+2W(x,y) + a_y - a_x. \qedhere
\end{align*}
\end{proof}

\paragraph*{Swapping vertices.}
To turn a minimal bisection $\S$ into a good bisection $\T=(T,\notT)$, we need at least one vertex in $T$ 
with high deficiency. One way to increase the deficiency of a vertex $u\in S$ is to move vertices
that are not adjacent to $u$ away from $S$ 
and to bring the same number of vertices that are adjacent to $u$ into $S$.
We define the {\em rank} of a vertex $u$ with respect to bisection $\S$ as
$$
 \rank_\S(u)= \left\lceil\frac{a_u+1-\defi_\S(u)}{2}\right\rceil.
$$
Note that a vertex $u$ of $\rank_\S(u)$ has deficiency $\defi_\S(u)$ such that
\begin{equation}
\label{eq:def_rank}
a_u - 2\rank_\S(u) +1 \leq \defi_\S(u)\leq a_u-2\rank_\S(u)+2.
\end{equation}
It is not hard to see that the rank is exactly the number of vertices that need to be moved. We next 
formalize the notion of swapping of vertices and prove that it is always possible to increase the deficiency 
of a non-stubborn vertex $x$ to $a_x$. 

Given a bisection $\S = (S, \notS)$ and a vertex $u$,
a \emph{$u$-pair} for $\S$ is a pair of sets $(A_u,B_u)$ such that:
\begin{itemize}
\item if $u\in S$, then $A_u\subseteq S\cap\notN(u)$ and $B_u\subseteq\notS\cap N(u)$ with $|A_u|=|B_u|=\rank_\S(u)$;
\item if $u\in\notS$, then $A_u \subseteq S \cap N(u)$ and $B_u \subseteq \notS \cap \notN(u)$ with $|A_u|=\rank_\S(u)$ and $|B_u|=\rank_\S(u)-1$.
\end{itemize}
The bisection $\T$ \emph{associated} with the $u$-pair $(A_u,B_u)$ for $\S$ is defined as
\begin{itemize}
\item if $u\in S$,    $\T=(S\setminus A_u\cup B_u,\notS\setminus B_u\cup A_u)$;
\item if $u\in\notS$, $\T=(\notS\setminus B_u\cup A_u,S\setminus A_u\cup B_u)$.
\end{itemize}
The next lemma shows that a $u$ is a good vertex in the bisection associated with a $u$-pair.
\begin{lemma}
\label{lem:goodvertex}
For each bisection $\S$, let $u$ be a vertex of the graph, $(A_u, B_u)$ a $u$-pair for $\S$, and $\T$ the bisection associated to $(A_u, B_u)$.
Then $\defi_\T(u)\geq a_u+1$.
\end{lemma}
\begin{proof}
Denote $\rank_\S(u)$ by $\ell$ and thus $\defi_\S(u) \geq a_u -2\ell +1$.
If $u \in S$, then, by definition of $u$-pair, $W(u, A_u) = 0$ and $W(u,B_u) = \ell$.
Hence,
\begin{align*}
 \defi_\T(u) & = W(u,T) - W(u, \notT)\\
 & = W(u,S) - W(u, \notS) - 2W(u,A_u) + 2W(u, B_u)\\
 & = \defi_\S(u) + 2\ell \geq a_u +1.
\end{align*}
 If $u \in \notS$, then, by definition of $u$-pair, $W(u, A_u) = \ell$ and $W(u,B_u) = 0$.
Hence,
\begin{align*}
 \defi_\T(u) & = W(u,T) - W(u, \notT)\\
 & = W(u,\notS) - W(u, S) + 2W(u,A_u) - 2W(u, B_u)\\
 & = \defi_\S(u) + 2\ell \geq a_u +1. \qedhere
\end{align*}
\end{proof}

Next lemma proves that, for every bisection $\S$ and every vertex $u$, 
a $u$-pair for $\S$ exists if and only if vertex $u$ is non-stubborn.
\begin{lemma}\label{lem:construction}
For every bisection $\S=(S,\notS)$ and every vertex $u$, a $u$-pair for $\S$ exists if and only if $u$ is non-stubborn.
\end{lemma}
\begin{proof}
Suppose that $u$ is a stubborn vertex and let $(A_u, B_u)$ a $u$-pair for $\S$.
By Lemma~\ref{lem:goodvertex}, in the bisection $\T$ associated with this $u$-pair,
$\defi_\T(u) \geq a_u + 1$. But this contradicts Lemma~\ref{lem:defiForbibben}.

Consider now a non-stubborn $u$. Let us denote $\rank_\S(u)$ by $\ell$. 
If $u \in S$, then it is sufficient to show that 
$$W(u,S)\leq \frac{n+1}{2}-\ell\qquad\text{and}\qquad W(u,\notS)\geq\ell.$$
Indeed, we have 
$$    W(u,S)=\frac{d(u)+\defi_\S(u)}{2}\qquad\text{and}\qquad 
  W(u,\notS)=\frac{d(u)-\defi_\S(u)}{2}.$$
Since $u$ is non-stubborn, we have $a_u+1\leq d(u)\leq n-a_u-2$.
Moreover, recall that $\defi_\S(u)\leq a_u-2\ell+2$.
Whence
\begin{align*}
 W(u,S) & \leq \frac{(n-a_u-2)+(a_u+2-2\cdot \ell)}{2}=
        \frac{n-2\cdot\ell}{2}<\frac{n+1}{2}-\ell\\
 W(u,\notS) & \geq \frac{(a_u+1)-(a_u+2-2\cdot\ell)}{2}=\ell - \frac{1}{2}.
\end{align*}
Since $W(u,\notS)$ is an integer, then it must be the case that $W(u,\notS) \geq \ell$.

If $u \in \notS$, we instead need to show that
$$W(u,S)\geq \ell\qquad\text{and}\qquad W(u,\notS)\leq \frac{n-1}{2} - \ell + 1.$$
A reasoning similar to the one above proves that these inequalities hold.
\end{proof}

Hence, if $\T$ is the bisection associated to $u$-pair $(A_u, B_u)$ for $\S$,
then $u$ is certainly a good vertex for $\T$. Thus, if $\T$ is not good then there is 
a vertex $y$ that is an obstruction for $\T$.
In the last case, we will say that the vertex $u$,
the $u$-pair $(A_u, B_u)$ and the bisection $\T$ are \emph{obstructed} by $y$.

\paragraph*{Stubborn vertices cannot be obstructions.}
Next lemma says that stubborn vertices are sort of neutral: they cannot be obstruction but they cannot be good
either.
\begin{lemma}
\label{lem:defiForbibben}
For every bisection $\S=(S,\notS)$ and every stubborn vertex $x\in S$ it holds that
$$-a_x \leq\defi_\S(x) \leq a_x.$$
\end{lemma}
\begin{proof}
By definition, 
$x$ is stubborn if $a_x\geq \min\left\{d(x), n-d(x)-1\right\}$.
The statement is obvious if $a_x \geq d(x)$. Thus consider $x\in S$ with $d(x)\geq a_x\geq n-d(x)-1$.
From the definition of $\defi_\S(x)$ and from the fact that
$d(x)=W(x,S)+W(x,\notS)$
we obtain
$$W(x,S)=\frac{d(x)+\defi_\S(x)}{2}\geq \frac{n-a_x-1+\defi_\S(x)}{2}$$
Since $W(x,S)\leq\frac{n+1}{2}-1$, we obtain $\defi_\S(x)\leq a_x$.

Similarly, we have 
$$\frac{n-1}{2}\geq W(x,\notS)=\frac{d(x)-\defi_\S(x)}{2}\geq\frac{n-a_x-1-\defi_\S(x)}{2}$$
whence we obtain $\defi_\S(x)\geq -a_x$.
\end{proof}

\section{Main theorem}
Our main result is the following.
\begin{theorem}
Every graph $G$ with an odd number of vertices and at least one non-stubborn vertex
has a subvertable belief assignment $\blf$ and a swing vertex $u$ for $\blf$.
Moreover, $\blf$ and $u$ can computed in polynomial time.
\end{theorem}
We prove the theorem by exhibiting a 
polynomial-time algorithm (see Algorithm~\ref{algo}) that, given a graph $G$ with an odd number of vertices,
and at least one of which that is non-stubborn, 
returns a good bisection $\S$ and a good vertex $u$ for $\S$.
The theorem then follows from Lemma~\ref{lem:good_bisection}.

\input{algo}

First, we note that the algorithm runs in time that is polynomial on the size of the input.
Indeed, a bisection of $3$-minimal potential at Line~\ref{line:bisection}
can be efficiently computed through a local search algorithm \cite{SY91SIAM},
and all remaining steps only involve computationally easy tasks.

Next we prove that the algorithm is correct; that is, it outputs $(\T, u)$
where $\T$ is a good bisection and $u$ is a good vertex for $\T$.
Recall that, by Lemma~\ref{lem:defiForbibben}, it is sufficient to check that
$\defi_\T(u) \geq a_u + 1$ and that non-stubborn vertices $x \in S$ have $\defi_\T(x) \geq -a_x$.

\subsection{Warm-up Cases}
In this section we show that if Algorithm~\ref{algo} stops before reaching 
Line~\ref{line:second_part} 
then it returns a good bisection and a good vertex.

\paragraph*{The algorithm stops at Line~\ref{line:wuc1}.}
In this case, $\T=(\notS\cup\{u\},S\setminus\{u\})$ and 
$\defi_\S(u)\leq -a_u-1$.
Since $u\in S$, we have that 
$\defi_\T(u)=-\defi_\S(u)\geq a_u+1$.
Moreover, for every non-stubborn $x \in T$, $x \neq u$,
$\defi_\T(x) = \defi_\S(x) + 2 W(x,u).$
By applying Lemma~\ref{lem:minimality} to $u\in S$ and $x\in \notS$, 
we obtain 
$\defi_\S(x)+2W(x,u)\geq -\defi_\S(u)+a_u-a_x.$
Therefore 
\begin{align*}
\defi_\T(x)&=\defi_\S(x) + 2 W(x,u)\\
        &\geq -\defi_\S(u)+a_u-a_x \\
        &\geq 2 a_u+1-a_x\geq -a_x.
\end{align*}

\paragraph*{The algorithm stops at Line~\ref{line:wuc2}.}
In this case all vertices $x\in S$ have $\defi_\S(x)\geq -a_x$
(for otherwise the algorithm would have stopped at Line~\ref{line:wuc1})
and $u$ is a good vertex.

\paragraph*{The algorithm stops at Line~\ref{line:wuc3}.}
In this case we have that $\T=(\notS\cup\{w\},S\setminus\{w\})$ and
it must be the case that $-a_x \leq \defi_\S(x) \leq a_x$ for every $x \in S$
(for otherwise the algorithm would have stopped at an earlier step)
and there is $u \in \notS$ with $\defi_\S(u)\geq a_u+1$.

Now, $\defi_\T(w) = - \defi_\S(w) \geq - a_w$ and 
$\defi_\T(u) = \defi_\S(u) + 2W(u,w) \geq a_u + 1$.
Moreover, for every non-stubborn vertex $y \in T$, we have $\defi_\T(y)=\defi_\S(y)+2W(y,w)$.
By  applying Lemma~\ref{lem:minimality} to 
$w\in S$ and to $y \in \notS$, we obtain 
that $\defi_\S(y)+2W(y,w)\geq -\defi_\S(w)+a_w-a_y$ and therefore we can write
$$
\defi_\T(y)=\defi_\S(y)+2W(y,w) \geq -\defi_\S(w)+a_w-a_y \geq -a_y,
$$
since we have $\defi_\S(w)\leq a_w$.

\paragraph*{The algorithm stops at Line~\ref{line:neg_def}.}
In this case we have that $\T=(\notS\cup \{u\}\setminus A_u \cup B_u, S \setminus \{u\} \cup A_u \setminus B_u)$ and 
let us denote $\rank_\S(u)$ by $\ell$.

Since $\T$ is the bisection associated to a $u$-pair,
then, by Lemma~\ref{lem:goodvertex}, $\defi_\T(u) \geq a_u + 1$.
Moreover, 
$\defi_{\S'}(u)=-\defi_{\S}(u)$ 
(where $\S'$ is defined at Line~\ref{line:Sp_neg_def}) and therefore
\begin{align*}
\rank_{\S'}(u) & = \left\lceil \frac{a_u+1-\defi_{\S'}(u)}{2}\right\rceil = 
                   \left\lceil \frac{a_u+1+\defi_{\S}(u)}{2}\right\rceil\\
   &=\left\lceil\frac{a_u+1-\defi_\S(u)}{2}\right\rceil+\defi_{\S}(u)=\rank_{\S}(u) + \defi_{\S}(u) \leq \ell - 1.
\end{align*}

For every non-stubborn $x\in T$ with $x\ne u$ we have that 
$$
  \defi_\T(x) = W(x,\notS) - W(x,S) + 2W(x, u) - 2W(x,A_u) + 2W(x,B_u).
$$

If $x\in \notS\setminus A_u$, then 
\begin{align*}
\defi_\T(x)&= \defi_\S(x)+2W(x,u) - 2W(x,A_u)+2W(x,B_u)\\
           &\geq \defi_\S(x) +2 W(x,u) - 2|A_u| = \defi_\S(x) + 2 W(x,u) - 2\rank_{\S'}(u)\\
           & \geq \defi_\S(x) + 2 W(x,u) - 2\ell + 2.
\end{align*}
By applying Lemma~\ref{lem:minimality} to 
$x \in \notS$ and $u \in S$, we obtain 
$\defi_\S(x) + 2W(x,u) \geq -\defi_\S(u) + a_u - a_x$ and, 
from \eqref{eq:def_rank}, we obtain
$\defi_\S(u)\leq a_x-2\ell+2$.
Hence $\defi_\T(x) \geq -a_x$.

Finally, let us consider $x\in B_u$. Then $x \in S$ and we have 
\begin{align*}
\defi_\T(x)&= -\defi_\S(x)+2W(x,u)-2W(x,A_u)+2W(x,B_u)\\
           &\geq -\defi_\S(x)-2W(x,A_u) \geq -\defi_\S(x)-2\ell+2.
\end{align*}
However, by hypothesis $u$ has minimum $\rank_\S(u)$ among the non-stubborn vertices 
and thus it must be the case that $\rank_\S(x)\geq\rank_\S(u)$ which implies that
 $\defi_\S(x)\leq a_x+2-2\ell$.
Therefore, $\defi_\T(x)\geq -a_x$.

\subsection{Properties of the obstructions}
Most of the work in the remaining cases will be devoted to deal with obstructions.
Therefore, before to proceed, we prove some their useful properties.
\begin{lemma}
\label{lem:obstruct_prop}
Let $u \in \notS$ be a vertex of minimum rank for the bisection $\S$ and let $y$ be an obstruction for $u$. Then $y \in \notS$.
Similarly, let $u \in S$ be a vertex of minimum rank for the bisection $\S$
and assume there is no vertex of minimum rank in $\notS$. If $y$ is an obstruction for $u$, then $y \in S$.
\end{lemma}
\begin{proof}
Suppose $y$ is an obstruction for the bisection $\T$ associated to $u$-pair $(A_u,B_u)$ and let $\rank_\S(u) = \ell$.  

We proceed by contradiction.
Suppose first that $y \in S$ and $u \in \notS$.
Since $u$ is a vertex of minimum rank, it must be the case that 
$\rank_\S(y)\geq\ell$ and thus, by \eqref{eq:def_rank},  
$\defi_{\S}(y)\leq a_y + 2 - 2\ell$.
Then, it holds that
\begin{align*}
\defi_{\T}(y)&= W(y,\notS)-W(y,S)+2W(y,A_u)-2W(y,B_u)\\
           &\geq -\defi_\S(y)-2W(y,B_u)\\
           &\geq -\defi_\S(y)-2\cdot(\ell-1) \geq -a_y.
\end{align*}
This is a contradiction, because $y$ is an obstruction for $\T$ and $\defi_\T(y)<-a_y$.

Suppose now that $y \in \notS$ and $u \in S$.
Since there are no vertices of minimum rank in $\notS$, then
$\rank_\S(y)\geq\ell+1$ and thus, by \eqref{eq:def_rank},  
$\defi_{\S}(y)\leq a_y - 2\ell$.
Then, it holds that
\begin{align*}
\defi_{\T}(y)&= W(y,S)-W(y,\notS)-2W(y,A_u)+2W(y,B_u)\\
           &\geq -\defi_\S(y)-2W(y,A_u)\\
           &\geq -\defi_\S(y)-2\ell \geq -a_y.
\end{align*}
As above, this contradicts that $y$ is an obstruction for $\T$.
\end{proof}

\begin{lemma}
\label{lem:obstruction}
Let $\S$ be a bisection and let $u$ be a vertex of minimum rank in $\notS$.
Let $\T$ be the bisection associated with a $u$-pair $(A_u,B_u)$ for $\S$.
If vertex $y$ is an obstruction for $\T$, then $$\defi_\S(y)\leq -a_y+2\rank_\S(u)-3.$$
Moreover, for every non-stubborn $v\in S$ if
\begin{equation}
\label{eq:min_cond}
\defi_\S(v) + \defi_\S(y) + 2 W(v,y) \geq a_v - a_y, 
\end{equation}
then $v$ and $y$ are adjacent, $v$ has minimum rank
and
$$\defi_\S(y)\geq -a_y+2\rank_\S(u)-4.$$
\end{lemma}
\begin{proof}
Let $\ell$ be the minimum rank with respect to the bisection $\S$. 

Since $u\in \notS$,
by Lemma~\ref{lem:obstruct_prop}, $y \in \notS$. Moreover, its deficiency is such that
$$
 -a_y > \defi_{\T}(y) = \defi_\S(y)+2W(y,A_u)-2W(y,B_u) \geq \defi_\S(y)-2|B_u| = \defi_\S(y) -2(\ell - 1),
$$
from which we obtain that 
$$
\defi_\S(y) \leq -a_y+2\ell-3.
$$

For every non-stubborn vertex $v\in S$ that satisfies \eqref{eq:min_cond},
we have that $\rank_\S(v) \geq \ell$, and thus, by \eqref{eq:def_rank}, $\defi_\S(v) < a_v - 2\ell + 3$. Then
\begin{align*}
-a_y +2\ell -3 \geq \defi_\S(y) & \geq -\defi_\S(v) - 2W(v,y) + a_v - a_y \\
            & > -a_y + 2\ell - 3 - 2W(v,y),
\end{align*}
from which we obtain that $W(v, y) = 1$ and $\defi_\S(y) \geq -a_y + 2\ell - 4$.
Moreover, from Lemma~\ref{lem:minimality},
\begin{align*}
\defi_\S(v) & \geq - \defi_\S(y) - 2W(v,y) + a_v - a_y \\
            & \geq a_y - 2\ell + 3 - 2 + a_v - a_y     \\
            & =    a_v - 2\ell + 1,
\end{align*}
and this implies that $\rank_\S(v) \leq \ell$. But, since $\ell$ is the minimum rank with respect to $\S$, then $\rank_\S(v) = \ell$. 
\end{proof}

\begin{lemma}
\label{lem:obstruction2}
Let $\S$ be a bisection and suppose that there is no vertex in $\notS$ with minimum rank.
Let $u$ be a vertex of minimum rank in $S$.
Let $\T$ be the bisection associated with a $u$-pair $(A_u,B_u)$ for $\S$.
Suppose there is an obstruction $y$ for $\T$ with $\defi_\S(y) < 0$ and $\rank_\S(y) > \rank_\S(u)$.
Let $\S' = (\notS \cup \{y\}, S \setminus \{y\})$.
Then $\rank_{\S'}(y) \leq \rank_\S(u)$.
\end{lemma}
\begin{proof}
Let $\rank_\S(u) = \ell$.
Therefore $\rank_\S(y)>\ell$ and thus $\defi_\S(y)\leq a_y-2\ell$.
Moreover, since by Lemma~\ref{lem:obstruct_prop}, $\defi_{\S'}(y)=-\defi_{\S}(y)$ and thus
\begin{equation}\label{eq:lastcase}
 \rank_{\S'}(y) = \left\lceil \frac{a_y + 1 - \defi_{\S'}(y)}{2}\right\rceil = \left\lceil \frac{a_y + 1 + \defi_{\S}(y)}{2}\right\rceil \leq \frac{a_y + \defi_\S(y)}{2} + 1.
\end{equation}
Since $y$ is an obstruction for $\T$
$$
-a_y>\defi_{\T}(y)=\defi_\S(y)-2W(y,A_u)+2W(y,B_u) \geq \defi_\S(y) - 2\ell,
$$
where we used that $W(y, A_u) \leq |A_u| = \ell$.
Hence, $a_y + \defi_\S(y) < 2\ell$ and, by plugging this in \eqref{eq:lastcase},
we obtain $\rank_{\S'}(y) \leq\ell$. 
\end{proof}

\begin{lemma}
\label{lem:obstruction3}
Let $\S$ be a bisection and let $u$ be a vertex of minimum rank in $S$.
Let $\T$ be the bisection associated with a $u$-pair $(A_u,B_u)$ for $\S$.
Suppose there is an obstruction $y$ for $\T$ with $\defi_\S(y) \geq 0$.
Then $y$ has minimum rank $\ell=\left\lceil\frac{a_y+1}{2}\right\rceil$ and $\defi_\S(y) = 0$.
\end{lemma}
\begin{proof}
Let $\rank_\S(u) = \ell$.
We start by observing that
\begin{align*}
  \defi_{\T}(y)&=    \defi_\S(y) - 2W(y,A_u) + 2W(y,B_u)\\
                &\geq \defi_\S(y) - 2\ell.
\end{align*}
Since $y$ is an obstruction for $\T$, it must be that $\defi_{\T}\leq -(a_y+1)$ and 
thus, since $\defi_\S(y)\geq 0$, we have that $\ell\geq\frac{a_y+1}{2}$.

On the other side, since $\rank_\S(y)\geq\ell$, we have that
\begin{equation}\label{eq:defiy}
0 \leq \defi_\S(y) \leq a_y + 2 - 2\ell,
\end{equation}
from which we obtain the $\frac{a_y + 1}{2} \leq \ell \leq \frac{a_y}{2} + 1$.
Since $\ell$ is an integer, it follows that $\ell = \left\lceil\frac{a_y+1}{2}\right\rceil$.
 
Now, if $a_y$ is even (and thus $\ell = \frac{a_y}{2} + 1$),
then, by \eqref{eq:defiy}, we obtain $\defi_\S(y) = 0$.
Moreover, since $0 \geq a_y-2\ell+1$, $y$ has minimum rank $\ell$.

If $a_y$ is odd (and thus $\ell = \frac{a_y+1}{2}$),
then, by \eqref{eq:defiy}, we obtain $0\leq\defi_\S(y)\leq 1$.
If $\defi_\S(y) = 1$ then 
 $$\defi_{\T}(y) = \defi_\S(y) - 2W(y,A_v) + 2W(y, B_v)\geq 1 - 2\ell = -a_y, $$
which contradicts the fact that $y$ is an obstruction for $\T$.
Hence $\defi_\S(y) = 0$ and, since $0 \geq a_y-2\ell+1$, $y$ has minimum rank. 
\end{proof}

\subsection{\texorpdfstring{There is a non-stubborn vertex of minimal rank in $\notS$}{There is a non-stubborn vertex of minimal rank in notS}}
Consider now the case that there is a non-stubborn vertex $u \in \notS$ of minimum rank.
We then execute the procedure {\MinRankInNotS} (see Algorithm~\ref{algo_notS}).
\input{algo_notS}
If {\MinRankInNotS} stops at Line~\ref{line:case_notS_good_if},
then clearly the returned bisection is good and, 
since the bisection is associated with a $u$-pair,
then, by Lemma~\ref{lem:goodvertex}, $u$ is a good vertex for it.

\paragraph{{\tt MinRankInNotS} stops at Line~\ref{line:case_notS_notgood_notstubbornS}.}
In this case, $y$ is an obstruction for a vertex $u\in\notS$ of minimum rank, and 
thus, by Lemma~\ref{lem:obstruct_prop}, $y$ belongs to $\notS$.
Moreover, there is at least one non-stubborn vertex in $S$.
Observe that, since $\S$ has minimal potential, from Lemma~\ref{lem:minimality}
it follows that for every non-stubborn vertex $x \in S$
$$
 \defi_\S(x) + \defi_\S(y) + 2W(x,y) \geq a_x - a_y.
$$
Thus, by Lemma~\ref{lem:obstruction},
$y$ and $v$ are adjacent,
where $v$ is the vertex considered at Line~\ref{line:pair_notS}.
Therefore it is possible to pick a $v$-pair $(A_v,B_v)$ with $y\in B_v$.

Since $\T$ is the bisection associated with a $v$-pair,
by Lemma~\ref{lem:goodvertex}, $\defi_\T(v) \geq a_v + 1$.
Thus, we only need to prove that $\defi_\T(x) \geq -a_x$ for every non-stubborn $x \in T$.
First note that, by Lemma~\ref{lem:obstruction}, $v$ has minimum rank and we denote 
$\rank_\S(v)$ by $\ell$.
Moreover, for all non-stubborn vertices $x\in T$ we have
$$\defi_\T(x)=W(x,S)-W(x,\notS)-2W(x,A_v)+2W(x,B_v).$$
For non-stubborn $x\in B_v$, we have $\defi_\S(x)=W(x,\notS)-W(x,S)$ and thus
$$\defi_\T(x)=-\defi_\S(x)-2W(x,A_v)+2W(x,B_v)\geq -\defi_\S(x)-2\ell$$
We next prove that $\defi_\S(x)\leq a_x-2\ell$ and thus $\defi_\T(x)\geq -a_x$.
Suppose by contradiction that $\defi_\S(x)\geq a_x-2\ell+1$.
Therefore, by \eqref{eq:def_rank}, $x$ is a non-stubborn vertex of $\notS$ of minimum rank  $\ell$ 
and thus the for-loop starting at Line~\ref{line_for_notS} has considered
a bisection $\T'$ associated with an $x$-pair $(A_x, B_x)$ for $\S$ with $v\in A_x$ that admitted an obstruction $y'$.
Note that, by Lemma~\ref{lem:obstruct_prop}, $y' \in \notS$.
Then
\begin{align*}
 \defi_{\T'}(y') & = W(y', \notS) - W(y', S) + 2W(y', A_x) - 2W(y', B_x)\\
		 & \geq \defi_\S(y') + 2W(y',v) -2|B_x|.
\end{align*}
Since $|B_x|=\ell-1$ and, by Lemma~\ref{lem:obstruction},
$\defi_\S(y')\geq -a_{y'}+2\ell-4$ and $W(y',v) = 1$, we have that
$\defi_{\T'}(y') \geq -a_{y'}$.
This is a contradiction, because $y'$ is an obstruction for $\T'$.

We conclude the proof by considering non-stubborn $x\in S\setminus A_v$, $x\ne v$.
For such a vertex we have 
$$\defi_\T(x)=\defi_\S(x)-2W(x,A_v)+2W(x,B_v)\geq \defi_\S(x)-2\ell+2$$
where we used that, according to Lemma~\ref{lem:obstruction}, $x$ and $y$ are adjacent and, by construction, $y\in B_v$.
It is thus sufficient that $\defi_\S(x)\geq -(a_x-2\ell+2)$.
Note that, by Lemma~\ref{lem:obstruction}, $\rank_\S(x)=\ell$ and thus minimal.
Hence, it must be the case that $\defi_\S(x) \geq 0$ (otherwise the algorithm stops at Line~\ref{line:neg_def})
and $\defi_\S(x)\leq a_x-2\ell+2$.
Thus $a_x-2\ell+2\geq 0$ and then $\defi_\S(x)\geq-(a_x-2\ell+2)$, as desired.

\paragraph{{\tt MinRankInNotS} stops at Line~\ref{line:case_notS_notgood_stubbornS}.}
Since the algorithm has not stopped before reaching this line, 
then in $\notS$ there is a non-stubborn vertex of minimal rank,
and every vertex of $S$ is stubborn.
Before proving that $\T$ is a good bisection, we show that
there exists at least one vertex $w \in S$ that is not adjacent 
to obstruction $y$ defined at Line~\ref{line:obstruction_notS} and 
thus $\T$ can be constructed.

We proceed by contradiction and assume that $y$ is adjacent 
to every vertex of $S$.
Recall that, by Lemma~\ref{lem:obstruct_prop}, $y \in \notS$.
Then, by applying Lemma~\ref{lem:minimality} to any $z\in S$ and to $y$
we obtain that
$$\defi_\S(y)\geq -\defi_\S(z)-2W(z,y)+a_z - a_y\geq -a_y-2,$$
where we used that $z$ is stubborn and thus, 
by Lemma~\ref{lem:defiForbibben}, $\defi_\S(z)\leq a_z$.
Remember that $y$ is an obstruction for $u$-pair $(A_u,B_u)$ where $u \in \notS$ 
is a vertex of minimum rank $\ell$. Let us denote by $\T'$ the bisection
associated with this $u$-pair. 
We have that
\begin{align*}
\defi_{\T'}(y)&=\defi_\S(y) + 2W(y,A_v) - 2W(y,B_v)\\
           &\geq -a_y-2 + 2W(y,A_v)-2W(y,B_v)\\
           &\geq -a_y-2+2W(y,A_v)-2(\ell-1)
\end{align*}
Since $y$ is adjacent to all vertices of $S$ then 
$W(y, A_v)=|A_v|= \ell$, and thus $\defi_{\T'}(y) \geq -a_y$.
This is a contradiction, as $y$ is an obstruction for $\T'$.
We have thus established that $y$ is not adjacent to all the vertices in $S$
and thus the algorithm can pick vertex $w$ at Line~\ref{line:Sp_notS}.

Since $\T$ is the bisection associated to a $y$-pair,
then, by Lemma~\ref{lem:goodvertex}, $\defi_\T(y) \geq a_y + 1$.
Thus, we only need to prove that $\defi_\T(x) \geq -a_x$ 
for every non-stubborn $x\in T$.
Consider the bisection $\S'$ defined at Line~\ref{line:Sp_notS}.
Note that $\defi_{\S'}(y) = - \defi_{\S}(y)$ and
$$\rank_{\S'}(y) = \left\lceil\frac{a_y+1-\defi_{\S'}(y)}{2}\right\rceil = \left\lceil\frac{a_y+1+\defi_{\S}(y)}{2}\right\rceil \leq \ell - 1,$$
where in the last inequality we used that, by Lemma~\ref{lem:obstruction}, $\defi_\S(y)\leq -a_y+2\ell-3$.

Let $(A_y, B_y)$ be the $y$-pair defined at Line~\ref{line:y_pair_notS}.
Since all vertices in $S$ are stubborn, then the only non-stubborn vertices in $T$
different from $y$ belong to $B_y \subseteq \notS \setminus \{y\} \cup \{w\}$.
For every such vertex $x$, we have that,
by minimality of $\ell$, $\defi_\S(x)\leq a_x-2\ell+2$,
$W(x,y) = 1$. Moreover, $W(x,w)\leq 1$ and $|A_y| = \rank_{\S'}(y)$.
Then
\begin{align*}
\defi_\T(x)&= W(x, S') - W(y, S') + 2W(x,B_y)-2W(x,A_y)\\
	   &= W(x,S)-W(x,\notS)+2W(x,y)-2W(x,w)+2W(x,B_y)-2W(x,A_y)\\
           &\geq -\defi_\S(x)-2|A_y| \geq -a_x.\qedhere
\end{align*}

\begin{remark}
\label{remark:notSmin}
 The procedure {\MinRankInNotS} correctly returns a pair $(\T, u)$,
 where $\T$ is a good bisection and $u$ is its good vertex,
 whenever the bisection $\S$ in input is a minimal bisection (not necessarily $3$-minimal)
 with a vertex in $\notS$ of minimum rank.
\end{remark}

We leverage on this property of {\MinRankInNotS} in the next section.

\subsection{\texorpdfstring{There is no non-stubborn vertex of minimal rank in $\notS$}{There is no non-stubborn vertex of minimal rank in notS}}
Finally, let us consider the case in which the algorithm invokes procedure {\MinRankInS} 
(described as Algorithm~\ref{algo_S}).
In this case, all non-stubborn vertices of minimum rank belong to $S$. 
Moreover, all such vertices have non-negative deficiency for otherwise the 
Algorithm would have stopped at Line~\ref{line:neg_def}.

\input{algo_S}

Clearly, if {\MinRankInS} stops at Line~\ref{line:case_S_good_if}, 
Line~\ref{line:returnT1},
Line~\ref{line:returnT2},
 Line~\ref{line:returnT4},
Line~\ref{line:returnT5} or Line~\ref{line:returnT7},
then the bisection output is good and $u$ is a good vertex for it.

Suppose now that {\MinRankInS} stops 
at Line~\ref{line:case_S_notgood_if},
Line~\ref{line:returnT3},
Line~\ref{line:neg_defS},
Line~\ref{line:returnT9},
Line~\ref{line:lastreturn}
or at Line~\ref{line:returnT8}.
Since in all cases the algorithm returns a pair $(\T,v)$ where $\T$ is the bisection associated to a $v$-pair,
then, by Lemma~\ref{lem:goodvertex}, $\defi_\T(v) \geq a_v + 1$.
Thus, we only need to prove that $\defi_\T(x) \geq -a_x$ for every 
non-stubborn $x \in T$.

\subsubsection{{\tt MinRankInS} stops at Line~\ref{line:case_S_notgood_if}}
In this case, we have that $u$ is a vertex of $S$ with minimum rank $\ell$. 
Vertex $y$ is an obstruction of bisection $\T$ associated with $u$-pair $(A_u, B_u)$,
and $\defi_\S(y)<0$. By Lemma~\ref{lem:obstruct_prop}, $y \in S$.
Observe that $\rank_\S(y)>\ell$,
for otherwise Algorithm~\ref{algo} would have stopped at Line~\ref{line:neg_def}.
From Lemma~\ref{lem:obstruction2}, we obtain that $\rank_{\S_0}(y)\leq\ell$.
We remind the reader that $\S_0=(\notS\cup\{y\}, S \setminus \{y\})$ (see Line~\ref{line:Sp_S_if}) and 
$\T_0=(\notS\cup\{y\}\setminus A_y\cup B_y, S\setminus\{y\}\cup A_y\setminus B_y)$.
 
For every non-stubborn $x\in T_0\setminus\{y\}$, we have that 
$$\defi_\T(x)= W(x,\notS)-W(x,S)+2W(x,y)-2W(x,A_y)+2W(x,B_y).$$
If $x\in\notS\setminus A_y$ then
 \begin{align*}
 \defi_\T(x)&= \defi_\S(x)+2W(x,y) - 2W(x,A_y)+2W(x,B_y)\\
           &\geq \defi_\S(x) +2 W(x,y) - 2|A_y|\\
           &= \defi_\S(x) + 2 W(x,y) - 2\rank_{\S'}(y)\\
           & \geq \defi_\S(x) + 2 W(x,y) - 2\ell.
\end{align*}
Since $\rank_\S(y) > \ell$, by \eqref{eq:def_rank}, $\defi_\S(y) \leq a_y - 2\ell$.
By applying Lemma~\ref{lem:minimality} 
to $y\in S$ and $x \in \notS$ we obtain that
$$\defi_\S(x) + 2W(x,y) \geq -\defi_\S(y) + a_y - a_x \geq -a_x + 2\ell.$$
Hence $\defi_\T(x) \geq -a_x$.

Finally, if $x\in B_y$, then $x \in S$ and, by definition of $y$-pair, $W(x,y) = 1$.
Therefore we have 
\begin{align*}
\defi_\T(x)&= -\defi_\S(x)+2W(x,y)-2W(x,A_y)+2W(x,B_y)\\
           &\geq -\defi_\S(x)-2W(x,A_y) +2 \geq -\defi_\S(x)-2\rank_{\S'}(y)+2 \geq -\defi_\S(x)-2\ell+2
\end{align*}
Since $\ell$ is the minimum rank, it must be the case that $\rank_\S(x)\geq\ell$ which implies that
 $\defi_\S(x)\leq a_x+2-2\ell$.
Therefore, $\defi_\T(x)\geq -a_x$.

\subsubsection{{\tt MinRankInS} reaches Line~\ref{line:obstruction}}
We remind the reader that in this case 
$u \in S$ is a non-stubborn vertex of minimum rank $\ell$ and 
$y$ is an obstruction to bisection $\T$ associated with $u$-pair $(A_u,B_u)$ for $\S$.
By Lemma~\ref{lem:obstruct_prop}, $y \in S$.
Note also that $\defi_\S(y) \geq 0$,
for  otherwise {\tt MinRankInS} would have stopped at Line~\ref{line:case_S_notgood_if}.
From Lemma~\ref{lem:obstruction3}, it then follows that
$\defi_\S(y) = 0$ and $\rank_\S(y)= \ell = \left\lceil\frac{a_y + 1}{2}\right\rceil$.
Moreover, given $y$-pair $(A_y, B_y)$ of $\S_1=(\notS\cup\{y\},S\setminus\{y\})$ (see Line~\ref{line:T1_pair} of {\MinRankInS}),
$y_1$ is either a vertex of $\left(\notS \cup \{y\} \setminus A_y\right) \cap \notN(y)$ with $\defi_\S(w) = a_y - a_{y_1}$
(see Line~\ref{line:obstruction_y1_alt} of {\tt MinRankInS})
or it is an obstruction to bisection $\T_1$ associated with this pair (see Line~\ref{line:obstruction_y1} of {\tt MinRankInS}).
Note that, by Lemma~\ref{lem:obstruct_prop}, even in this last case $y_1 \in \notS \cup \{y\} \setminus A_y$.

\paragraph{Properties of $y$ and $y_1$.}
Before proving that the bisections returned by {\tt MinRankInS} after Line~\ref{line:obstruction} are good,
we need to establish some properties of $y$ and $y_1$.

\begin{lemma}
 \label{lem:second}
 $W(y,y_1) = 0$.
\end{lemma}
\begin{proof}
This is obvious if $y_1$ has been defined at Line~\ref{line:obstruction_y1_alt} of {\MinRankInS}.

Suppose instead that $y_1$ is an obstruction to $\T_1$.
Then we have that $\defi_{\T_1}(y_1)<-a_{y_1}$. 
On the other hand
$$\defi_{\T_1}(y_1)\geq\defi_{\S_1}(y_1)-2|B_y|\geq\defi_{\S_1}(y_1)- 2\rank_{\S_1}(y)
  =\defi_{\S_1}(y_1)- 2\ell,$$
where we used that $\defi_{\S_1}(y) = - \defi_\S(y) = 0$ and thus $\rank_{\S_1}(y) = \left\lceil\frac{a_y + 1}{2}\right\rceil = \ell$.

We then obtain that $\defi_{\S_1}(y_1)\leq - a_{y_1} + 2\ell-1$ and 
\begin{equation}
\label{eq:defS_y1}
\defi_\S(y_1)=\defi_{\S_1}(y_1)-2W(y_1,y)\leq -a_{y_1}-1+2\ell - 2 W(y_1,y).
\end{equation}
Suppose now, for sake of contradiction, that $W(y,y_1) = 1$ and thus
\begin{equation}
\label{eq:ub_y1}
\defi_\S(y_1)\leq -a_{y_1}-3+2\ell.
\end{equation}

At Line~\ref{line:u_pair_S_if} of {\MinRankInS}, since $y_1\in\notS\cap N(y)$,
a $y$-pair $(A_y, B_y)$ for $\S$ such that $y_1 \in B_y$ has been considered
and the bisection $\T'$ associated with this $y$-pair 
had an obstruction that we call $y'$. It must be the case that $\defi_\S(y')\geq 0$,
for otherwise the procedure would have stopped at Line~\ref{line:obstruction_S_if}.
Then, from Lemma~\ref{lem:obstruct_prop}, $y'\in S$ and, from Lemma~\ref{lem:obstruction3}, 
$y'$ is a vertex of minimum rank $\ell$ in $\S$ and $\defi_\S(y')=0$. 
This implies that
$\ell=\rank_\S(y')=\left\lceil\frac{a_{y'}+1}{2}\right\rceil\leq\frac{a_{y'} + 2}{2}$ and 
thus $a_{y'}\geq 2\ell -2$.
Moreover, as we shall show next, $y_1$ and $y'$ are neighbors. 
Indeed, from \eqref{eq:ub_y1} and by applying Lemma~\ref{lem:minimality} to $y_1\in\notS$ and $y'\in S$,
we have
\begin{align*}
2\ell-a_{y_1}-3 & \geq  \defi_\S(y_1) \geq -\defi_\S(y') - 2 W(y_1,y') +a_{y'}-a_{y_1}  \\
                & \geq  2\ell-a_{y_1}-2-2W(y_1,y').
\end{align*}
From the above chain of inequalities we obtain $W(y_1,y')=1$.
Since $y'$ is an obstruction to $\T'$, then $\defi_{\T'}(y')<-a_{y'}$.
On the other hand, 
\begin{align*}
\defi_{\T'}(y')&\geq\defi_{\S}(y')-2W(y',A_y)+2W(y', B_y)\\
&\geq -2\ell + 2W(y',y_1)\\
&\geq -(a_{y'} + 2) + 2 = - a_{y'},
\end{align*}
that is a contradiction.
Therefore we can conclude that $y$ and $y_1$ are not neighbors. 
\end{proof}

\begin{lemma}
 \label{lem:third}
 If $\defi_\S(y_1)\neq a_y-a_{y_1}$, then $a_y$ is even and $\defi_\S(y_1) = a_y-a_{y_1} + 1$. 
\end{lemma}
\begin{proof}
By applying Lemma~\ref{lem:minimality} to $y\in S$ and $y_1\in\notS$, we obtain
\[\defi_\S(y_1)\geq-\defi_\S(y)-2W(y,y_1)+ a_y-a_{y_1} = a_y - a_{y_1}.\]
Since $\defi_\S(y_1)\neq a_y-a_{y_1}$, it must be then the case that $\defi_\S(y_1)\geq a_y - a_{y_1} + 1$.

On the other side, by substituting $W(y,y_1)=0$ in \eqref{eq:defS_y1}, we obtain that 
$$\defi_{\S}(y_1)\leq -a_{y_1}+2\ell-1=2\left\lceil\frac{a_y+1}{2}\right\rceil-(a_{y_1}+1) \leq a_y - a_{y_1} + 1$$,
from which the claim follows.
\end{proof}

\begin{lemma}
\label{lem:fourth}
 If $\defi_\S(y_1)\neq a_y-a_{y_1}$, then $\defi_\S(w) \geq a_y-a_{w} + 1 - W(w, y)$ for every $w \in \notS$.
\end{lemma}
\begin{proof}
If $\defi_\S(y_1)\neq a_y-a_{y_1}$, then, by Lemma~\ref{lem:third},
we have that $a_y$ is even and $\defi_\S(y_1) = a_y-a_{y_1} + 1$.
We first show that $\defi_\S(y_1) = a_y-a_{y_1} + 1$
implies that $\defi_\S(w) \geq a_y-a_{w} + 1$ for every $w \in \notS \cap \notN(y)$.

Since $y_1$ is an obstruction to $\T_1$, then $y_1 \notin A_y$.
By our choice of $A_y$ (see Line~\ref{line:T1_pair}), it must be then the case that
$\defi_\S(w) = a_y-a_{w} + 1$ for every $w \in A_y$.

Consider now $w \in \left(\notS \setminus A_w\right) \cap \notN(y)$.
By applying Lemma~\ref{lem:minimality} to $y\in S$ and $w\in\notS$, we 
obtain
$$\defi_\S(w)\geq-\defi_\S(y)-2W(y,w)+ a_y-a_{w} = a_y - a_{w},$$
where we used that $\defi_\S(y) = 0$ and $w \in \notN(y)$.
However, it cannot be the case that $\defi_\S(w) = a_y - a_{w}$,
otherwise $y_1$ was returned at Line~\ref{line:obstruction_y1_alt},
and thus $\defi_\S(y_1) = a_y-a_{y_1}$, a contradiction.

Now we prove that if $a_y$ is even, then $\defi_\S(w) \geq a_y-a_{w}$ for every $w \in \notS \cap N(y)$.

We show that there is a vertex $x \in S$ such that $\defi_\S(x) = 0$,
$a_{x}=a_y$ and $W(w,x) = 0$.
Then, by applying Lemma~\ref{lem:minimality} to $x\in S$ and $w\in\notS$, we 
obtain
$$\defi_\S(w)\geq-\defi_\S(x)-2W(x,w)+ a_{x}-a_{w} = a_y - a_{w}.$$

Suppose indeed, by sake of contradiction, that $w$ is a neighbor of every vertex $x \in S$ such that $\defi_\S(x) = 0$ and $a_{x}=a_y$.
Observe that at Line~\ref{line:u_pair_S_if}, since $w\in\notS\cap N(y)$,
a $y$-pair $(A_y, B_y)$ for $\S$ such that $w \in B_y$ has been considered
and the bisection $\T'$ associated with this $y$-pair 
had an obstruction that we call $y'$. It must be the case that $\defi_\S(y')\geq 0$,
for otherwise the procedure would have stopped at Line~\ref{line:obstruction_S_if} of {\MinRankInS}.
Then, from Lemma~\ref{lem:obstruct_prop}, $y'\in S$ and, from Lemma~\ref{lem:obstruction3}, 
$y'$ is a vertex of minimum rank $\ell$ in $\S$, $\defi_\S(y')=0$ and $\ell = \left\lceil\frac{a_{y'} + 1}{2}\right\rceil$.
Moreover $a_{y'}$ is even, because otherwise $y$ was returned at Line~\ref{line:obstruction_y_alt},
that is not possible since, by hypothesis, $a_y$ is odd.
This implies that
$\frac{a_{y'} + 2}{2} = \ell=\frac{a_{y} + 2}{2}$ and 
thus $a_{y'}=a_y$.
Finally, since $W(y', w) = 1$, then, by construction of $B_y$, $W(y', B_y) \geq 1$.

Now, since $y'$ is an obstruction to $\T'$, then $\defi_{\T'}(y')<-a_{y'}$.
On the other hand, 
\begin{align*}
\defi_{\T'}(y')&\geq\defi_{\S}(y')-2W(y',A_y)+2W(y', B_y)\\
&\geq -2\ell + 2 = -(a_{y'} + 2) + 2 = - a_{y'},
\end{align*}
that is a contradiction. 
\end{proof}

\begin{lemma}
 \label{lem:sixth}
 $\rank_{\S_2}(y) = \rank_{\S_2}(y_1) = \ell$.
\end{lemma}
\begin{proof}
We remind the reader that $\S_2=(S\cup \{y_1\} \setminus \{y\}, \notS \cup\{y\} \setminus \{y_1\})$. 
Since, by Lemma~\ref{lem:second}, $W(y, y_1) = 0$, we have $\defi_{\S_2}(y)=-\defi_\S(y)=0$.
  Hence
  $$
    \rank_{\S_2}(y) = \left\lceil \frac{a_{y} + 1 - \defi_{\S_2}(y)}{2}\right\rceil = \left\lceil \frac{a_{y} + 1}{2}\right\rceil = \ell.
  $$
  Similarly, $\defi_{\S_2}(y_1) = -\defi_\S(y_1)$ and thus
  $$
   \rank_{\S_2}(y) = \left\lceil \frac{a_{y_1} + 1 - \defi_{\S_2}(y_1)}{2}\right\rceil = \left\lceil \frac{a_{y_1} + 1 + \defi_{\S}(y_1)}{2}\right\rceil.
  $$
  If $\defi_{\S}(y_1) = a_y - a_{y_1}$, then $\rank_{\S_2}(y) = \left\lceil \frac{a_{y} + 1}{2}\right\rceil = \ell$.
  If $\defi_{\S}(y_1) = a_y - a_{y_1} + 1$ (and thus, by Lemma~\ref{lem:third}, $a_y$ is even), then
  $\rank_{\S_2}(y) = \left\lceil \frac{a_{y} + 2}{2}\right\rceil = \ell$.
\end{proof}
  
\begin{lemma}
 \label{lem:last}
 For every $u \in S_2$ and every $v \in \notS_2 \setminus \{y\}$, we have
\begin{itemize}
 \item $\defi_{\S_2}(u) + \defi_{\S_2}(v) + 2W(u,v) \geq a_u - a_v$, if $\defi_\S(y_1) = a_y - a_{y_1}$ or $u = y_1$;
 \item $\defi_{\S_2}(u) + \defi_{\S_2}(v) + 2W(u,v) = a_u - a_v + c + 2W(y,v) + 2W(u,y_1) - 2W(u,y) - 2W(v, y_1)$, 
 for $c \geq \max\left\{0, 2W(u,v) - 2W(y_1,u) - W(y,v)\right\}$, otherwise.
\end{itemize}
\end{lemma}
\begin{proof}
If $\defi_\S(y_1) = a_y - a_{y_1}$, then 
\begin{align*}
 \Phi(\S_2) - \Phi(\S) & = \defi_\S(y) + \defi_\S(y_1) + 2W(y,y_1) + a_{y_1} - a_y\\
   & = 0 + (a_y - a_{y_1}) + 0 + a_{y_1} - (a_y - 1) = 0,
\end{align*}
and thus $\Phi(\S_2) = \Phi(\S)$,
and, since $\S$ has $3$-minimal potential, then $\S_2$ has minimal potential.
The desired property then follows from Lemma~\ref{lem:minimality}.

Assume now that $\defi_\S(y_1) \neq a_y - a_{y_1}$.
Then, by Lemma~\ref{lem:third}, $a_y$ is even and $\defi_\S(y_1) = a_y - a_{y_1} + 1$.
We first consider the case that $u = y_1$. Then,
\begin{align*}
 \defi_{\S_2}(y_1) + \defi_{\S_2}(v) + 2W(y_1,v) & = - \defi_{\S}(y_1) + \left(\defi_{\S}(v) - 2W(y_1,v) + 2W(y,v)\right) + 2W(y_1,v)\\
 & = - (a_y - a_{y_1} + 1) + \defi_{\S}(v) + 2W(y,v)
\end{align*}
Since, by Lemma~\ref{lem:fourth}, $\defi_\S(v) \geq a_y - a_v + 1 - W(y, v)$, we have
\begin{align*}
 \defi_{\S_2}(y_1) + \defi_{\S_2}(v) + 2W(y_1,v) & \geq -(a_y - a_{y_1} + 1) + (a_y - a_v + 1 - W(y,v)) + 2W(y,v)\\
 & = a_{y_1} - a_v + W(y, v) \geq a_{y_1} - a_v.
\end{align*}

Consider now the case that $u \neq y_1$.
Observe that $\defi_{\S_2}(u) = \defi_\S(u) + 2W(u, y_1) - 2W(u,y)$ and $\defi_{\S_2}(v) = \defi_\S(v) - 2W(v, y_1) + 2W(v,y)$.
Moreover, by applying Lemma~\ref{lem:minimality} to $u \in S$ and $y_1 \in \notS$, we have that there is $c_u \geq 0$ such that
\begin{equation}
\label{eq:defiS}
 \defi_\S(u) = -\defi_\S(y_1) - 2W(y_1,u) + a_u - a_{y_1} + c_u = a_u - a_y - 1 + c_u - 2W(y_1,u).
\end{equation}

Moreover, by Lemma~\ref{lem:fourth}, there is $c_v \geq 0$ such that $\defi_\S(v) = a_y - a_v + 1 + c_v - W(y,v)$.
By applying Lemma~\ref{lem:minimality} to $u \in S$ and $v \in \notS$, we have that
\begin{align*}
 0 & \leq \defi_\S(u) + \defi_\S(v) + 2W(u,v) + a_v - a_u\\
 & = (a_u - a_y - 1 + c_u - 2W(y_1,u)) + (a_y - a_v + 1 + c_v - W(y,v)) + 2W(u,v) + a_v - a_u\\
 & = (c_u + c_v) - 2W(y_1,u) - W(y,v) + 2W(u,v).
\end{align*}
By setting $c = (c_u + c_v) - 2W(y_1,u) - W(y,v) + 2W(u,v) \geq 0$, we than have
\begin{align*}
 &\defi_{\S_2}(u) + \defi_{\S_2}(v) + 2W(u,v)\\
 & \qquad = (\defi_\S(u) + 2W(u, y_1) - 2W(u,y)) + (\defi_\S(v) - 2W(v, y_1) + 2W(v,y)) + 2W(u,v)\\
 & \qquad = a_u - a_v + c + 2W(y,v) + 2W(u,y_1) - 2W(u,y) - 2W(v, y_1). \qedhere
\end{align*}
\end{proof}

\begin{lemma}
 \label{lem:caseT3noproblem}
 For every $u \in S \setminus \{y_1\}$ and $v \in \notS \setminus \{y\}$,
 if $W(u, y) = 1$ and $W(u, y_1) = 0$,
 then $\defi_{\S_2}(u) + \defi_{\S_2}(v) + 2W(u,v) \geq a_u - a_v - 1$.
\end{lemma}
\begin{proof}
The claim follows from Lemma~\ref{lem:last} if $\defi_\S(y_1) = a_y - a_{y_1}$.

Consider instead the case that $\defi_\S(y_1) = a_y - a_{y_1} - 1$. We first observe that 
 \begin{align*}
 \Phi(\S_2) - \Phi(\S) & = \defi_\S(y) + \defi_\S(y_1) + 2W(y,y_1) + a_{y_1} - a_y\\
 & = 0 + (a_y - a_{y_1} + 1) + 0 + a_{y_1} - a_y = 1,
\end{align*}
 Consider now the bisection $\S' = (S_2 \cup v \setminus u, \notS_2 \cup u \setminus v)$.
 It must be the case that $\Phi(\S') \geq \Phi(S_2) - 1$,
 otherwise $\Phi(\S') < \Phi(\S)$, contradicting the minimality of $\S$.
 Then,
 $$
  - 1 \leq \Phi(\S') - \Phi(\S_2) = \defi_{\S_2}(u) + \defi_{\S_2}(v) + 2W(u,v) + a_{v} - a_u,
 $$
 from which the claim follows.
\end{proof}

\paragraph{{\tt MinRankInS} stops at Line~\ref{line:returnT3}.}
Then there is a vertex $v \in S_2$, whose rank in $\S_2$ is less than $\rank_{\S_2}(y) = \ell$.
Note that, since $\defi_{\S_2}(v) = \defi_\S(v) + 2W(v,y_1) - 2W(v,y)$,
\begin{align*}
 \rank_{\S_2}(v) & = \left\lceil \frac{a_v + 1 - \defi_{\S_2}(v)}{2}\right\rceil = \left\lceil \frac{a_v + 1 - \defi_{\S}(v) - 2W(v,y_1) + 2W(v,y)}{2}\right\rceil\\
 & = \left\lceil \frac{a_v + 1 - \defi_{\S}(v)}{2}\right\rceil  - W(v,y_1) + W(v,y) = \rank_\S(v) - W(v,y_1) + W(v,y).
\end{align*}
Hence, $\rank_{\S_2}(v) < \ell$ if and only if $\rank_\S(v) = \ell$ (that is, $v$ has minimum rank in $\S$), $W(v,y_1) = 1$ and $W(v,y) = 0$.
From this we obtain that for every vertex $v$ with $\rank_{\S_2}(v) < \ell$,
it holds that $\defi_\S(v) \geq 0$
(since $v$ has minimum rank in $S$ and no vertex of minimum rank in $S$ with negative deficiency can exist,
otherwise a good bisection was returned at Line~\ref{line:neg_def} of Algorithm~\ref{algo}),
and, $\defi_{\S_2}(v) \geq 2$.
We also observe that every vertex $x \in \notS_2 = \notS \cup \{y\} \setminus \{y_1\}$ has $\rank_{\S_2}(x) \geq \ell$.
If $x=y$, then this follows from Lemma~\ref{lem:sixth}.
If $x \neq y$, then the claim follows since $\rank_\S{x} \geq \ell + 1$,
and the rank can decrease of at most one when two vertices are swapped.

The bisection $\T_2$ associated to $v$-pair $(A_v, B_v)$ for $\S_2$ has an obstruction $y_2$.
By Lemma~\ref{lem:obstruct_prop}, $y_2 \in S_2 \setminus A_v$.
Suppose that $\defi_{\S_2}(y_2) \geq 0$, then, from Lemma~\ref{lem:obstruction3},
it follows that $\defi_{\S_2}(y_2) = 0$ and has minimum rank, i.e., $\rank_{\S_2}(y_2) = \ell - 1$.
However, this is a contradiction, since we showed that if $\rank_{\S_2}(y_2) = \ell - 1$, then $\defi_{\S_2}(y_2) \geq 2$.

It must be then the case that $\defi_{\S_2}(y_2) < 0$ and clearly $\rank_{\S_2}(y_2) \geq \ell$
(i.e., $y_2$ has not minimum rank in $S_2$).
Then, by Lemma~\ref{lem:obstruction2}, we have that $\rank_{\S_3}(y_2) \leq \ell -1$,
where $\S_3 = (\notS_2 \cup \{y_2\}, S_2 \cup \{y_2\})$ (see Line~\ref{line:S3}).
Note that it must be also the case that either $\rank_\S(y_2) \geq \ell + 1$ or
$\rank_\S(y_2) = \ell$, $W(y_2, y) = 1$ and $W(y_2, y_1) = 0$.
Indeed $\rank_\S(y_2)\geq \ell$, since $\ell$ is the minimum rank in $\S$.
If $\rank_\S(y_2) = \ell$, then $\defi_{\S}(y_2) \geq 0$.
In this case if $W(y_2, y) = 0$ and $W(y_2, y_1) = 1$, then $\rank_{\S_2}(y_2) = \ell - 1$, a contradiction.
If $W(y_2, y) = W(y_2, y_1)$, then $\defi_{\S_2}(y_2) = \defi_{\S}(y_2) \geq 0$, still a contradiction.

We are now ready to prove that the bisection $\T_3$ returned at Line~\ref{line:returnT3} is good.
Recall that $\T_3$ is the bisection associated to $y_2$-pair $(A_{y_2}, B_{y_2})$ for $\S_3$,
i.e., $\T_3 = (\notS_2 \cup \{y_2\} \setminus A_{y_2} \cup B_{y_2},
S_2 \setminus \{y_2\} \cup A_{y_2} \setminus B_{y_2})$,
where $|A_{y_2}| = |B_{y_2}| = \rank_{\S_3}(y_2) \leq \ell - 1$.

We first prove that for every $x \in \notS_2 \cup \{y_2\} \setminus A_{y_2}$,
we have that $\defi_{\T_3}(x) \geq -a_x$.
If $x \neq y$, we distinguish two cases.
If $\rank_\S(y_2) \geq \ell + 1$,
then by applying Lemma~\ref{lem:minimality} to $y_2 \in S$ and $x \in \notS$,
we have that
$$
 \defi_\S(x) + 2W(x,y_2) \geq - \defi_\S(y_2) + a_{y_2} - a_x \geq -a_x + 2\ell,
$$
where we used that $\rank_\S(y_2) \geq \ell + 1$ and thus $\defi_\S(y_2) \leq a_{y_2} - 2\ell$.
Then,
$$
 \defi_{\S_2}(x) = \defi_\S(x) + 2W(x,y) - 2W(x, y_1) \geq -a_x + 2\ell - 2.
$$
If $\rank_\S(y_2) < \ell + 1$,
then, as stated above, it must be the case that $\rank_\S(y) = \ell$, $W(y_2, y) = 1$ and $W(y_2, y_1) = 0$.
Then, from Lemma~\ref{lem:caseT3noproblem}, it holds that
\begin{align*}
 \defi_{\S_2}(x) + 2W(x,y_2) & \geq - \defi_{\S_2}(y_2) + a_{y_2} - a_x  - 1\\
 & \geq - a_x + 2\ell - 1,
\end{align*}
where we used that $\rank_{\S_2}(y_2) = \rank_\S(y_2) + 1 = \ell + 1$ and thus $\defi_\S(y_2) \leq a_{y_2} - 2\ell$.

Hence, in both cases, we have
\begin{align*}
 \defi_{\T_3}(x) & \geq \defi_{\S_3}(x) - 2W(x,A_{y_2}) \geq \defi_{\S_2}(x) + 2W(x,y_2) - 2(\ell-1)\\
 & \geq  - a_x + 2\ell - 1 - 2(\ell-1) \geq -a_x + 1.
\end{align*}
If $x = y$, then 
\begin{align*}
 \defi_{\T_3}(y) & \geq \defi_{\S_3}(y) - 2(\ell-1)\\
 & = \defi_{\S_2}(y) + 2W(y,y_2) - 2(\ell-1)\\
 & = - \defi_{\S}(y) + 2W(y,y_2) - 2(\ell-1),
\end{align*}
where we used that $\defi_{\S_2}(y) = - \defi_{\S}(y)$ because $W(y, y_1) = 0$.
Since, as showed above, $\defi_{\S}(y) = 0$ and $\ell = \left\lceil\frac{a_y + 1}{2}\right\rceil \leq \frac{a_y + 2}{2}$,
we have that $\defi_{\T_3}(y) \geq - a_y + 2W(y,y_2) \geq -a_y$.

Finally, we prove that for every $x \in B_{y_2}$, $\defi_{\T_3}(x) \geq -a_x$.
Recall that $B_{y_2} \subseteq S_2 \setminus \{y_2\}$
and $W(x, y_2) = 1$ for every $x \in B_{y_2}$.
We distinguish two cases. If $x \neq y_1$, then
\begin{align*}
 \defi_{\T_3}(x) & \geq - \defi_{\S_3}(x) - 2(\ell - 1) = - \defi_{\S_2}(x) + 2W(x,y_2) - 2(\ell-1)\\
 & = -\defi_{\S}(x) + 2W(x,y) - 2W(x,y_1) + 2 - 2(\ell-1)\\
 & \geq -a_x + 2\ell -2 + 2W(x,y) - 2W(x,y_1) + 2 - 2(\ell-1) \geq -a_x,
\end{align*}
where we used that $\rank_\S(x) \geq \ell$ and thus $\defi_\S(x) \leq a_x - 2\ell + 2$.
If $x = y_1$, then
\begin{align*}
 \defi_{\T_3}(y_1)
 & \geq - \defi_{\S_3}(y_1) - 2(\ell-1)\\
 & = - \defi_{\S_2}(y_1) + 2W(y_1,y_2) - 2(\ell-1)\\
 & = \defi_{\S}(y_1) + 2 - 2(\ell-1),
\end{align*}
where we used that $\defi_{\S_2}(y_1) = -\defi_{\S}(y_1)$ because $W(y, y_1) = 0$ and $W(y_1, y_2) = 1$ because $y_1 \in B_{y_2}$.
Since $\defi_\S(y_1) \geq a_y - a_{y_1} \geq 2(\ell-1) - a_{y_1}$, then
$\defi_{\T_3}(y_1) \geq - a_{y_1} + 2 \geq -a_{y_1}$.

\paragraph{{\tt MinRankInS} stops at Line~\ref{line:neg_defS}.}
In this case $y$ and $y_1$ have minimum rank in $\S_2$
and there is a vertex $w \in S \setminus \{y\} \cup \{y_1\}$ of minimum rank $\ell$ and negative deficiency in $\S_2$.
Note that if $\rank_{\S_2}(w) \leq \rank_{\S}(w)$, then $\defi_{\S_2}(w) \geq \defi_{\S}(w)$.
Thus, since in $\S$ all vertices of minimum rank have non-negative deficiency,
it must be the case that $\rank_\S(w) = \ell + 1$, and $W(w, y) = 0$ and $W(w,y_1)=1$.
Thus the $w$-pair defined at Line~\ref{line:pair_neg_defS} can be constructed.

Consider now the bisection $\S_4$ defined at Line~\ref{line:Sp_neg_defS}.
Observe that $\defi_{\S_4}(w)=-\defi_{\S_2}(w)$ and therefore
\begin{align*}
\rank_{\S_4}(w) & = \left\lceil \frac{a_w+1-\defi_{\S_4}(w)}{2}\right\rceil = 
                   \left\lceil \frac{a_w+1+\defi_{\S_4}(w)}{2}\right\rceil\\
   &=\left\lceil\frac{a_w+1-\defi_{\S_2}(w)}{2}\right\rceil+\defi_{\S_2}(w)=\rank_{\S_2}(w) + \defi_{\S_2}(w) \leq \ell - 1,
\end{align*}
where we used that $w$ has rank $\ell$ and negative deficiency in $\S_2$.

Now, for every $x \in \notS_2 \setminus A_w$, we have
$$
\defi_{\T_5}(x) \geq \defi_{\S_2}(x)+2W(x,w) -2\rank_{\S_4}(w) \geq -\defi_{\S_2}(w) + a_w - a_x - 2(\ell -1),
$$
where we used that, by Lemma~\ref{lem:last},
\begin{align*}
 \defi_{\S_2}(x)+2W(x,w) & \geq -\defi_{\S_2}(w) + a_w - a_x + 2W(y,x) + 2W(w,y_1) - 2W(w,y) - 2W(x, y_1)\\
 & \geq -\defi_{\S_2}(w) + a_w - a_x.
\end{align*}
Since $\rank_{\S_2}(w) = \ell$, then $\defi_{\S_2}(w) \leq a_w - 2\ell + 2$,
from which we achieve that $\defi_{\T_5}(x) \geq - a_x$.

Finally, let us consider $x\in B_w \subseteq S_2$. We have 
$$
\defi_{\T_5}(x) \geq -\defi_{\S_2}(x) -2\rank_{\S_4}(w) \geq -\defi_{\S_2}(x)-2(\ell - 1).
$$
However, by hypothesis $w$ has minimum rank among the non-stubborn vertices 
and thus it must be the case that $\rank_{\S_2}(x)\geq\rank_{\S_2}(w) = \ell$ which implies that
 $\defi_{\S_2}(x)\leq a_x - 2\ell + 2$.
Therefore, $\defi_\T(x)\geq -a_x$.

\subsubsection{{\tt MinRankInS} reaches Line~\ref{line:obstruction_y6}}
In this case $y$ and $y_1$ have minimum rank in $\S_2$,
but there is an obstruction $y_4$ to bisection $\T_4$
associated with $y$-pair $(A_y, B_y)$ for $\S_2$.
From Lemma~\ref{lem:obstruct_prop}, $y_4 \in \notS_2$.
Then, from Lemma~\ref{lem:last}, it holds that
$$
 \defi_{\S_2}(y_1) + \defi_{\S_2}(y_4) + 2W(y_1,y_4) \geq a_{y_1} - a_{y_4}.
$$
Thus, by Lemma~\ref{lem:obstruction}, $y_4$ and $y_1$ are adjacent.
This shows that it is always possible to pick a $y_1$-pair $(A_{y_1},B_{y_1})$ as defined at Line~\ref{line:T5}.
However, there is obstruction $y_6$ to the bisection $\T_6$ associated with this pair.

We start by proving some properties of $y_4$ and $y_6$ and of the bisection
$\S_6$ defined in Line~\ref{line:S5} and obtained by swapping $y_4$ and $y_6$.
In particular, we will prove that $\S_6$ is a minimal bisection.
Note that this implies that, according to Remark~\ref{remark:notSmin},
if the bisection is returned at Line~\ref{line:lastreturn},
then it enjoys the desired properties.
Hence, it will be sufficient to show
that the bisections returned at Line~\ref{line:returnT9} and Line~\ref{line:returnT8} are good.

\paragraph{Properties of $y_4$ and $y_6$.}
\begin{lemma}
 \label{lem:s6minimal}
 $\Phi(\S_6) = \Phi(\S_2) - 1 = \Phi(\S)$. Hence, since $\S$ is $3$-minimal, $\S_6$ is minimal.
\end{lemma}
\begin{proof}
We first remind the reader that 
\begin{align*}
 \Phi(\S_2) - \Phi(\S) & = \defi_\S(y) + \defi_\S(y_1) + 2W(y,y_1) + a_{y_1} - a_y\\
   & \geq 0 + (a_y - a_{y_1} + 1) + 0 + a_{y_1} - (a_y - 1) = 1,
\end{align*}
and thus $\Phi(\S_2) \geq \Phi(\S) + 1$. Moreover, since $\S$ is $3$-minimal, $\Phi(\S_6) \geq \Phi(\S)$.

Suppose now, by sake of contradiction, that $\Phi(\S_2) - 1 \neq \Phi(\S)$ or $\Phi(\S_6) \neq \Phi(\S_2)$.
In both cases we have that $\Phi(\S_6) \geq \Phi(\S_2)$ and thus
$$
 0 \leq \Phi(\S_6) - \Phi(\S_2) = \defi_{\S_2}(y_6) + \defi_{\S_2}(y_4) + 2W(y_6, y_4) + a_{y_4} - a_{y_6}.
$$
It is then possible to apply Lemma~\ref{lem:obstruction} to $y_6$ and $y_4$,
and have that these two vertices are adjacent. Thus,
since by construction $y_4\in B_{y_1}$, $W(y_6, B_{y_1}) \geq 1$.
Moreover, by Lemma~\ref{lem:obstruction}, $\rank_{\S_2}(y_6)=\ell$ and thus minimal.
Hence, it must be the case that $\defi_{\S_2}(y_6) \geq 0$ (otherwise the algorithm stops at Line~\ref{line:neg_defS})
and $\defi_\S(y_6)\leq a_{y_6}-2\ell+2$,
from which we achieve that $a_{y_6}-2\ell+2\geq 0$ and then $\defi_\S(y_6)\geq-(a_{y_6}-2\ell+2)$.
Therefore, 
$$
 \defi_{\T_6}(y_6)=\defi_{\S_2}(y_6)-2W(y_6,A_{y_1})+2W(y_6,B_{y_1})\geq -(a_{y_6}-2\ell+2) -2\ell +2 = -a_{y_6},
$$
that is a contradiction because $\defi_{\T_6}(y_6) < -a_{y_6}$ since $y_6$ is an obstruction to $\T_6$.
\end{proof}

By Lemma~\ref{lem:last}, since $\Phi(\S_6) = \Phi(\S_2) - 1$, it must be the case that $\defi_\S(y_1) \neq a_y - a_{y_1}$.
Then, from Lemma~\ref{lem:third}, we have that
$a_y$ is even (hence, $\ell = \frac{a_y + 2}{2}$) and $\defi_\S(y_1) = a_y - a_{y_1} + 1$.

\begin{lemma}
 \label{lem:conditions}
 One of the following conditions is satisfied:
\begin{enumerate}
 \item \label{item:cond1} $W(y_6, y) = W(y_4, y) = W(y_6, y_4) = 1$ and $W(y_1, y_6) = 0$;
 \item \label{item:cond2} $W(y_6, y) = 1$ and $W(y_4, y) = 0$;
 \item \label{item:cond3} $W(y_6, y_1) = W(y_6, y_4) = 0$.
\end{enumerate}
\end{lemma}
\begin{proof}
 By Lemma~\ref{lem:s6minimal} and Lemma~\ref{lem:last}, we have that
\begin{equation}
\label{eq:bad_case}
 \begin{aligned}
 - 1 = \Phi(\S_6) - \Phi(\S_2) & = \defi_{\S_2}(y_6) + \defi_{\S_2}(y_4) + 2W(y_6, y_4) + a_{y_4} - a_{y_6}\\
 & = c + 2W(y,y_4) + 2W(y_6,y_1) - 2W(y_6,y) - 2W(y_4, y_1)\\
 & = c + 2W(y,y_4) + 2W(y_6,y_1) - 2W(y_6,y) - 2
\end{aligned}
\end{equation}
for $c \geq \max\left\{0, 2W(y_6,y_4) - 2 W(y_1, y_6) - W(y, y_4)\right\}$.

Suppose, for sake of contradiction, that conditions \ref{item:cond1}-\ref{item:cond3} are not satisfied.
First consider the case that $W(y,y_4) = 0$. Then, $W(y_6,y) = 0$.
Now, if $W(y_6, y_1) = 1$, then \eqref{eq:bad_case} fails.
If $W(y_6,y_1) = 0$. Then, $W(y_6, y_4) = 1$.
Therefore, $c \geq 2$ and thus \eqref{eq:bad_case} fails.

Consider now that $W(y_4, y) = 1$. Then, either $W(y_6,y) = 0$ or $W(y_6, y_4) = 0$ or $W(y_6,y) = W(y_6, y_4) = W(y_1, y_6) = 1$.
In the first case and in the third case, \eqref{eq:bad_case} fails.
In the second case, we must have $W(y_6, y_1) = 1$, and then \eqref{eq:bad_case} fails again.
\end{proof}

\begin{lemma}
\label{lem:cx}
$c_{x} \leq 1 + 2W(y_1,x)$ for every $x \in S$.
\end{lemma}
\begin{proof}
If $c_{x} \geq 2 + 2W(y_1,x)$, then
$$
\defi_\S(x) = a_{x} - a_y - 1 + c_{y_6} - 2W(y_1,y_6) \geq a_{y_6} - a_y + 1.
$$
Therefore,
$$
 \rank_\S(y_6) = \left\lceil\frac{a_{y_6} + 1 - \defi_\S(y_6)}{2}\right\rceil = \left\lceil\frac{a_y}{2}\right\rceil = \ell - 1,
$$
where we used that $\ell = \frac{a_y + 2}{2}$.
Anyway, this is a contradiction since $\ell$ is the minimum rank in $\S$.
\end{proof}

\begin{lemma}
\label{lem:minrank}
 If $W(y_6, y) = 1$ and $W(y_4, y) = 0$, then there is a vertex in $\notS_6$ of minimum rank.
\end{lemma}
\begin{proof}
 Note indeed that
 $$\defi_{\S_6}(y) = \defi_{\S_2}(y) + 2W(y_6, y) - 2W(y_4, y) = 2,$$
 where we used the hypothesis and the fact that $\defi_{\S_2}(y) = 0$ since $W(y, y_1) = 0$.
 
 Hence,
 $$\rank_{\S_6}(y) = \left\lceil\frac{a_y + 1 - \defi_{\S_6}(y)}{2}\right\rceil = \left\lceil\frac{a_y + 1}{2}\right\rceil - 1 = \ell - 1.$$
 
 The claim then follows, by observing that for every $x \in S_2$, $x \neq y_6$,
 $\rank_{\S_2}(x) \geq \ell$ (otherwise a bisection would be returned at Line~\ref{line:returnT3}),
 and thus $\rank_{\S_6}(x) \geq \ell - 1$.
 
 Moreover for $x = y_4$, by Lemma~\ref{lem:fourth}, we have
\begin{align*}
\defi_{\S_2}(y_4) &= \defi_{\S}(y_4) - 2W(y_1,y_4) + 2W(y, y_4)\\
& \geq a_y - a_{y_4} + 1 + W(y, y_4) - 2W(y_1,y_4)
& \geq a_y - a_{y_4} - 1.
\end{align*}
Hence,
\begin{align*}
 \defi_{\S_6}(y_4) & = - \defi_{\S_2}(y_4) - 2W(y_4, y_6)\\
 & \leq a_{y_4} - a_y + 1\\
 & = a_{y_4} - 2\ell + 3.
\end{align*}
Then $\rank_{\S_6}(y_4) \geq \ell - 1$.
\end{proof}

\paragraph{{\tt MinRankInS} stops at Line~\ref{line:returnT9}.}
In this case we have that $W(y_6, y) = W(y_4, y) = W(y_6, y_4) = 1$ and $W(y_1, y_6)$.
Let $\S_5$ be the bisection defined at Line~\ref{line:S7} of {\tt MinRankInS},
i.e., $\S_5 = (\notS_2 \cup \{y_6\}, S_2 \setminus \{y_6\})$.
Observe that $\defi_{\S_5}(y) = \defi_{\S_2}(y) + 2W(y, y_6) = -\defi_\S(y) - 2W(y, y_1) + 2W(y,y_6) = 2$,
where we used that $\defi_\S(y) = 0$, $W(y, y_1) = 0$ from Lemma~\ref{lem:second}, and $W(y,y_6) = 1$.
Hence,
$$
 \rank_{\S_5}(y) = \left\lceil\frac{a_y + 1 - \defi_{\S_7}(y)}{2}\right\rceil = \left\lceil\frac{a_y + 1}{2}\right\rceil - 1 = \ell - 1,
$$
where we used that $\ell = \frac{a_y + 2}{2}$.

We now prove that the bisection $\T_7 = (\notS_2 \cup \{y_6\} \setminus A_y \cup B_y, S_2 \setminus \{y_6\} \cup A_y \setminus B_y)$
defined at Line~\ref{line:T9} is good,
that is, for every non-stubborn vertex in $x \in \notS_2 \cup \{y_6\} \setminus A_y \cup B_y$ it holds that $\defi_{\T_7}(x) \geq a_x$.
Consider first $x \in \notS_2 \setminus A_y = \notS \setminus \{y_1\} \cup \{y\} \setminus A_y$.
Observe that
\begin{align*}
 \defi_{\S_5}(x) & = \defi_{\S_2}(x) + 2W(x, y_6)\\
 & = \defi_\S(x) + 2W(x, y) - 2W(x,y_1) + 2W(x, y_6).
\end{align*}
Since, from Lemma~\ref{lem:fourth}, $\defi_\S(x) \geq a_y - a_x + 1 - W(x, y)$,
we have that
\begin{align*}
 \defi_{\S_5}(x) & \geq a_y - a_x + 1 + W(x, y) - 2W(x,y_1) + 2W(x, y_6)\\
 & \geq 2\ell - a_x - 1 + W(x, y) - 2W(x,y_1) + 2W(x, y_6),
\end{align*}
where we used that $\ell = \frac{a_y + 2}{2}$ and thus $a_y \geq 2\ell - 2$.
Therefore,
\begin{align*}
 \defi_{\T_7}(x) & \geq \defi_{\S_5}(x) - 2\rank_{\S_5}(y)\\
 & \geq 2\ell - a_x - 1 + W(x, y) - 2W(x,y_1) + 2W(x, y_6) - 2\ell + 2\\
 & = - a_x + 1 + W(x, y) - 2W(x,y_1) + 2W(x, y_6)
\end{align*}
We next show that either $W(x,y_1) = 0$ or $W(x,y) + W(x, y_6) \geq 1$,
from which it follows that $\defi_{\T_7}(x) \geq -a_x + 1$.

Suppose by sake of contradiction that $W(x,y_1) = 1$ and $W(x,y) + W(x, y_6) = 0$.
Then $\defi_{\S_2}(x) = a_y - a_x + 1 + W(x, y) - 2W(x,y_1) = a_y - a_x - 1$.
On the other side, from \eqref{eq:defiS} we have
$$\defi_{\S_2}(y_6) = a_{y_6} - a_y - 1 + c_{y_6} - 2W(y_6,y) = a_{y_6} - a_y - 3 + c_{y_6}.$$
Let $\S' = (S_2 \setminus \{y_6\} \cup \{x\}, \notS_2 \setminus \{x\} \cup \{y_6\})$.
Since $\S$ is $2$-minimal and, by Lemma~\ref{lem:s6minimal}, $\Phi(\S_2) = \Phi(\S) + 1$, we need that
$$
 -1 \leq \Phi(\S') - \Phi(\S_2) = \defi_{\S_2}(y_6) + \defi_{\S_2}(x) + 2W(y_6,x) - a_{y_6} + a_x = c_{y_6} - 4.
$$
However, this is a contradiction since, by Lemma~\ref{lem:cx}, $c_{y_6} \leq 1 + 2W(y_1,y_6) = 1$,
where we used that $W(y_1, y_6) = 0$ by hypothesis.

Consider now $x = y_6$.
Observe that, from \ref{eq:defiS} and $W(y_1, y_6) = 0$, it follows that
\begin{align*}
 \defi_{\S_5}(y_6) & = -\defi_{\S_2}(y_6) = -\defi_\S(y_6) + 2W(y_6, y) - 2W(y_6,y_1)\\
 & = -a_{y_6} + a_y + 3 - c_{y_6}\\
 & = -a_{y_6} + 2\ell + 1 - c_{y_6},
\end{align*}
where we used that $\ell = \frac{a_y + 2}{2}$.
Since, by Lemma~\ref{lem:cx}, $c_{y_6} \leq 1 + 2W(y_1, y_6) = 1$, we have
$$
 \defi_{\T_7}(y_6) = \defi_{\S_5}(y_6) - 2(\ell - 1) = -a_{y_6} + 3 - c_{y_6} \geq - a_{y_6} + 2.
$$

Finally, consider $x \in B_y$. Recall that in this case $W(x,y)=1$.
From \eqref{eq:defiS}, we achieve
\begin{align*}
 \defi_{\S_5}(x) & = \defi_{\S_2}(x) - 2W(x, y_6)\\
 & = \defi_\S(x) - 2 + 2W(x,y_1) - 2W(x, y_6)\\
 & = a_{x} - a_y - 3 + c_{x} - 2W(x, y_6)\\
 & = a_{x} - 2\ell - 1 + c_{x}  - 2W(x, y_6).
\end{align*}
Then, by Lemma~\ref{lem:cx}
$$
 \defi_{\T_7}(x) = -\defi_{\S_5}(y_6) - 2(\ell - 1) = -a_{x} + 3 - c_{x} + 2W(x, y_6) \geq - a_{x},
$$
where we used that $c_x \leq 1 + 2W(x,y_1) \leq 3$.

\paragraph{{\tt MinRankInS} stops at Line~\ref{line:returnT8}.}
We first note that if {\MinRankInS} reaches Line~\ref{line:returnT9},
then it must be the case that $W(y_6, y_1) = W(y_6, y_4) = 0$.
Indeed, it cannot be the case that $W(y_6, y) = W(y_4, y) = W(y_6, y_4) = 1$,
otherwise a bisection would be returned at Line~\ref{line:returnT9}.
And, it cannot be the case that $W(y_6, y) = 1$ and $W(y_4, y) = 0$,
otherwise, according to Lemma~\ref{lem:minrank}, there is a vertex of minimum rank in $\notS_6$,
and thus a bisection would be returned at Line~\ref{line:lastreturn}.
Then, by Lemma~\ref{lem:conditions}, we have $W(y_6, y_1) = W(y_6, y_4) = 0$.

Consider then $\S_6$ as defined at Line~\ref{line:S5},
i.e., $\S_6 = (S_6, \notS_6) = (S_2 \cup \{y_4\} \setminus \{y_6\}, \notS_2 \cup\{y_6\} \setminus \{y_4\})$.
First note that
$$\defi_{\S_6}(y_1) = \defi_{\S_2}(y_1) + 2W(y_1, y_4) - 2W(y_1, y_6) = \defi_{\S_2}(y_1) + 2,$$
where $W(y_1, y_4) = 1$, as follows from Lemma~\ref{lem:obstruction} when applied to $y_4$ and $y_1$ in $\S_2$.
Consequently $\rank_{\S_6}(y_1) = \rank_{\S_2}(y_1) - 1 = \ell -1$.

Moreover observe that every vertex $x \in S_6 \setminus \{y_1\}$ of minimum rank in $\S_6$ has positive deficiency.
Indeed, either $x \in S_2 \setminus \{y_1, y_6\}$ or $x = y_4$.
In the first case, $\rank_{\S_2}(x) \geq \ell$ (otherwise a bisection was returned at Line~\ref{line:returnT3}).
Thus, $\rank_{\S_6}(x) = \ell - 1$ if and only if $\rank_{\S_2}(x) = \ell$, $W(x, y_4) = 1$ and $W(x, y_6) = 0$.
However, if $\rank_{\S_2}(x) = \ell$, then $\defi_{\S_2}(x) \geq 0$
(otherwise a bisection was returned at Line~\ref{line:neg_defS}).
Then $\defi_{\S_6}(x) = \defi_{\S_2}(x) + 2W(x, y_4) - W(x, y_6) \geq 2$.

If instead $x = y_4$, then, by Lemma~\ref{lem:fourth}, there is $c_{y_4} \geq 0$ such that
\begin{align*}
\defi_{\S_2}(y_4) & = \defi_{\S}(y_4) - 2W(y_1,y_4) + 2W(y, y_4)\\
& = a_y - a_{y_4} + 1 + W(y, y_4) - 2W(y_1,y_4) + c_{y_4}.
\end{align*}
Hence,
\begin{align*}
 \defi_{\S_6}(y_4) & = - \defi_{\S_2}(y_4) - 2W(y_4, y_6)\\
 & = a_{y_4} - a_y - 1 - W(y, y_4) + 2W(y_1,y_4) - c_{y_4}.
\end{align*}
Then $\rank_{\S_6}(y_4) = \ell - 1$ if an only if $W(y, y_4) = c_{y_4} = 0$ and $W(y_1,y_4) = 1$.
In this case, $\defi_{\S_6}(y_4) = a_{y_4} - a_y + 1 = -\defi_\S(y_4)$.
The claim then follows by showing that $\defi_{\S}(y_4) < 0$.
Indeed, if $\defi_\S(y_4) \geq 0$, then $a_{y_4} \leq a_y + 1$.
Then, since $a_y$ is even,
 $$
  \rank_\S(y_4) = \left\lceil \frac{a_{y_4} + 1 - \defi_\S(y_4)}{2}\right\rceil \leq \left\lceil \frac{a_{y_4} + 1}{2}\right\rceil \leq \left\lceil \frac{a_{y} + 2}{2}\right\rceil = \ell.
 $$
 Hence, $y_4$ has minimum rank in $\S$, that is a contradiction.

 Finally, note that every vertex in $\notS_6$ has rank at least $\ell$,
 otherwise a bisection was returned at Line~\ref{line:lastreturn}.
 
 Consider now the bisection $\T_8$ defined at Line~\ref{line:T7}.
 Recall that $\T_8$ is the bisection associated with an $y_1$-pair $(A_{y_1}, B_{y_1})$ for $\S_6$.
 Since $\T_8$ has not been returned at Line~\ref{line:returnT7},
 then it has an obstruction $y_8$.
By Lemma~\ref{lem:obstruct_prop}, $y_8 \in S_6 \setminus A_{y_1}$.
If $\defi_{\S_6}(y_8) \geq 0$, then, from Lemma~\ref{lem:obstruction3},
it follows that $\defi_{\S_6}(y_8) = 0$ and has minimum rank.
However, this is a contradiction, since we showed that $\rank_{\S_6}(y_2) = \ell - 1$ implies $\defi_{\S_6}(y_2) > 0$.

It must be then the case that $\defi_{\S_6}(y_8) < 0$ and clearly $\rank_{\S_6}(y_8) \geq \ell$
(i.e., $y_8$ has not minimum rank in $\S_6$).
Then, by Lemma~\ref{lem:obstruction2}, we have that $\rank_{\S_7}(y_8) \leq \ell -1$,
where $\S_7 = (\notS_6 \cup \{y_8\}, S \setminus \{y_8\})$ has been defined at Line~\ref{line:S6}.

We are now ready to prove that the bisection $\T_9$ returned at Line~\ref{line:returnT8} is good.
Recall that $\T_9$ is the bisection associated to $y_8$-pair $(A_{y_8}, B_{y_8})$ for $\S_7$,
i.e., $\T_9 = (\notS_6 \cup \{y_8\} \setminus A_{y_8} \cup B_{y_8},
S_6 \setminus \{y_6\} \cup A_{y_6} \setminus B_{y_6})$,
where $|A_{y_8}| = |B_{y_8}| = \rank_{\S_7}(y_8) \leq \ell - 1$.

We first prove that for every $x \in \notS_6 \cup \{y_8\} \setminus A_{y_8}$,
we have that $\defi_{\T_9}(x) \geq -a_x$.
If $x \neq y_6$,
then, since by Lemma~\ref{lem:s6minimal}, $\S_6$ is minimal,
by applying Lemma~\ref{lem:minimality} to $y_8 \in S_6$ and $x \in \notS_6$,
we have that
\begin{align*}
 \defi_{\S_7}(x) & = \defi_{\S_6}(x) + 2W(x,y_8)\\
 & \geq - \defi_{\S_6}(y_8) + a_{y_8} - a_x\\
 & \geq -a_x + 2\ell - 2,
\end{align*}
where we used that $\rank_{\S_6}(y_8) \geq \ell$ and thus $\defi_{\S_6}(y_8) \leq a_{y_8} - 2\ell + 2$.
Therefore,
\begin{align*}
 \defi_{\T_9}(x) & \geq \defi_{\S_7}(x) - 2W(x,A_{y_8})\\
 & \geq -a_x + 2\ell - 2 - 2(\ell-1) = -a_x.
\end{align*}

If $x = y_6$, then, recalling that $W(y_6, y_4) = 0$ and thus $\defi_{\S_6}(y_6) = -\defi_{\S_2}(y_6)$, we have
\begin{align*}
 \defi_{\T_9}(y_6)
 & \geq \defi_{\S_7}(y_6) - 2(\ell-1)\\
 & = \defi_{\S_6}(y_6) + 2W(y_6,y_8) - 2(\ell-1)\\
 & = - \defi_{\S_2}(y_6) + 2W(y_6,y_8) - 2(\ell-1).
\end{align*}
By \eqref{eq:defiS} and by recalling that $W(y_6, y_1) = 0$,
\begin{align*}
 \defi_{\S_2}(y_6) & = \defi_\S(y_6) + 2W(y_6, y_1) - 2W(y_6, y)\\
 & \geq a_{y_6} - a_y - 1 - 2W(y_6, y)\\
 & = a_{y_6} - 2\ell + 1 - 2W(y_6, y),
\end{align*}
where we used that $\ell = \frac{a_y + 2}{2}$.
Hence,
\begin{align*}
 \defi_{\T_9}(y_6) & = -a_{y_6} + 2\ell - 1 + 2W(y_6, y) + 2W(y_6,y_8) - 2(\ell-1)\\
 & \geq -a_{y_6} + 1.
\end{align*}

Finally, we prove that for every $x \in B_{y_8}$, $\defi_{\T_9}(x) \geq -a_x$.
Recall that $B_{y_8} \subseteq S_6 \setminus \{y,6, y_8\}$
and $W(x, y_8) = 1$ for every $x \in B_{y_8}$.
We distinguish two cases. If $x \neq y_4$, then
\begin{align*}
 \defi_{\T_9}(x) & \geq - \defi_{\S_7}(x) - 2(\ell - 1) = - (\defi_{\S_6}(x) - 2W(x,y_8)) - 2(\ell-1)\\
 & = -(\defi_{\S_2}(x) + 2W(x,y_6) - 2W(x,y_4) - 2W(x,y_8)) - 2(\ell-1)\\
 & \geq -a_x + 2\ell - 2 - 2W(x,y_6) + 2W(x,y_4) + 2 - 2(\ell-1) \geq -a_x,
\end{align*}
where we used that $\rank_\S(x) \geq \ell$
and thus $\defi_{\S_2}(x) \leq a_x - 2\ell + 2$.

If $x = y_4$, then
\begin{align*}
 \defi_{\T_9}(y_4)
 & \geq - \defi_{\S_7}(y_4) - 2(\ell-1)\\
 & = - (\defi_{\S_6}(y_4) - 2W(y_4,y_8)) - 2(\ell-1)\\
 & = - (-\defi_{\S_2}(y_4) - 2W(y_4,y_8)) - 2(\ell-1)\\
 & = \defi_{\S_2}(y_4) + 2 - 2(\ell-1),
\end{align*}
where we used that $\defi_{\S_6}(y_4) = -\defi_{\S_6}(y_4)$ because $W(y_6, y_4) = 0$ and $W(y_4, y_8) = 1$ because $y_4 \in B_{y_8}$.
By Lemma~\ref{lem:fourth}, we have that
\begin{align*}
 \defi_{\S_2}(y_{4}) & = \defi_\S(y_4) + 2W(y_4, y) - 2W(y_4, y_1)\\
 & \geq a_y - a_{y_4} + 1 + W(y,y_4) - 2W(y_4, y_1)\\
 & = 2\ell - a_{y_4} - 1 + W(y,y_4) - 2W(y_4, y_1)
\end{align*}
Then,
\begin{align*}
 \defi_{\T_9}(y_4) & \geq 2\ell - a_{y_4} - 1 + W(y,y_4) - 2W(y_4, y_1) + 2 - 2(\ell-1)\\
 & \geq - a_{y_4} + 3 + W(y,y_4) - 2W(y_4, y_1)\\
 & \geq - a_{y_4} + 1.
\end{align*}

\section{Lower Bound}
We next show that deciding if it is possible to subvert the majority when starting from a weaker minority is a computationally hard problem, even if we start with a minority of size very close to $\frac{n-1}{2}$.
The main result of this section is given by the following theorem.
\begin{theorem}
\label{thm:reduction}
For every constant $0 < \varepsilon < \frac{133}{155}$, given a graph $G$ with $n$ vertices,
it is NP-hard to decide whether there exists a subvertable belief assignment
with at most $\frac{n-1}{2}(1-\varepsilon)$ vertices in the initial minority.
\end{theorem}
\begin{proof}
To prove our result, we will use a reduction from the NP-hard problem 2P2N-3SAT, that is, the problem of deciding whether a 3SAT formula in which every variable appears as positive in two clauses and as negative in two clauses has a truthful assignment or not (the NP-hardness of 2P2N-3SAT follows by the results in \cite{Y05TRA}).

Given an instance of 2P2N-3SAT, i.e. a Boolean formula $\phi$ with $C$ clauses and $V$ variables where $3C=4V$ (thus, $C$ is a multiple of $4$),
we will construct a graph $G(\phi)$ with odd $n$ vertices, such that $\phi$ has a satisfying assignment if and only if there exists a 
belief assignment $\blf$ for the vertices of $G(\phi)$, with at most $\frac{n-1}{2}(1-\varepsilon)$ vertices having belief $1$, and a sequence of updates leading from $\blf$ to an equilibrium in
which at least $\frac{n+1}{2}$ vertices adopt opinion $1$.

The graph $G(\phi)$ contains the following vertices and edges.
\begin{itemize}
\item For each variable $x$ of $\phi$,
$G(\phi)$ includes a {\em variable} gadget for $x$ consisting of $25$ 
vertices and $50$ edges (see Figure~\ref{fig:variable}).

The vertices of the variable gadget for $x$ are:
the {\em literal vertices}, $x$ and $\ox$;
vertices $v_1(x),\ldots,v_{7}(x)$, $v_1(\ox),\ldots,v_{7}(\ox)$, and $w_1(x), \ldots, w_{7}(x)$; 
vertex $v_0(x)$ and $w_0(x)$.

The edges of the variable gadget for $x$ are:
edges $(x,v_i(x))$, $(\overline{x},v_i(\overline{x}))$ and $(w_0(x), w_i(x))$, for $i=1, \ldots, 7$;
edges $(v_i(x),v_{i+1}(x))$ and $(v_i(\overline{x}),v_{i+1}(\overline{x}))$ for $i=1, \ldots, 6$;
edges $(v_{0}(x),v_7(x))$, $(v_{0}(x),v_7(\overline{x}))$, $(v_0(x), w_0(x))$,
$(w_0(x),v_i(x))$, $(w_0(x),v_i(\overline{x}))$.
\begin{figure}[htb]
 \begin{tikzpicture}[-,auto,on grid=true,semithick,
                     prof/.style={shape=circle,draw,inner sep=0pt,minimum size=1mm},
                     every label/.style={font=\tiny}]

  \node[prof] (V1) [label={below:$v_{1}(x)$}] {};
  \node[prof] (V2) [right=0.75cm of V1] [label={below:$v_{2}(x)$}] {};
  \node[prof] (V3) [right=0.75cm of V2] [label={below:$v_{3}(x)$}] {};
  \node[prof] (V4) [right=0.75cm of V3] [label={below:$v_{4}(x)$}] {};
  \node[prof] (V5) [right=0.75cm of V4] [label={below:$v_{5}(x)$}] {};
  \node[prof] (V6) [right=0.75cm of V5] [label={below:$v_{6}(x)$}] {};
  \node[prof] (V7) [right=0.75cm of V6] [label={below:$v_{7}(x)$}] {};
  \node[prof] (V0) [right=0.75cm of V7] [label={below:$v_{0}(x)$}] {};
  \node[prof] (V7N) [right=0.75cm of V0] [label={below:$v_{7}(\ox)$}] {};
  \node[prof] (V6N) [right=0.75cm of V7N] [label={below:$v_{6}(\ox)$}] {};
  \node[prof] (V5N) [right=0.75cm of V6N] [label={below:$v_{5}(\ox)$}] {};
  \node[prof] (V4N) [right=0.75cm of V5N] [label={below:$v_{4}(\ox)$}] {};
  \node[prof] (V3N) [right=0.75cm of V4N] [label={below:$v_{3}(\ox)$}] {};
  \node[prof] (V2N) [right=0.75cm of V3N] [label={below:$v_{2}(\ox)$}] {};
  \node[prof] (V1N) [right=0.75cm of V2N] [label={below:$v_{1}(\ox)$}] {};

  \node[prof] (W0) [above=1.25cm of V0] [label={right:$\qquad w_{0}(x)$}] {};

  \node[prof] (W4) [above=1.25cm of W0] [label=$w_{4}(x)$] {};
  \node[prof] (W3) [left=1cm of W4] [label=$w_{3}(x)$] {};
  \node[prof] (W2) [left=1cm of W3] [label=$w_{2}(x)$] {};
  \node[prof] (W1) [left=1cm of W2] [label=$w_{1}(x)$] {};
  \node[prof] (W5) [right=1cm of W4] [label=$w_{5}(x)$] {};
  \node[prof] (W6) [right=1cm of W5] [label=$w_{6}(x)$] {};
  \node[prof] (W7) [right=1cm of W6] [label=$w_{7}(x)$] {};

  \node[prof] (X) [below=1.25cm of V4] [label={left:$x \quad$}] {};
  \node[prof] (XN) [below=1.25cm of V4N] [label={right:$\quad \ox$}] {};

  \phantom{\node[prof] (j) [below=1.75cm of V3] {};}
  \phantom{\node[prof] (k) [below=1.75cm of V5] {};}
  \phantom{\node[prof] (l) [below=1.75cm of V3N] {};}
  \phantom{\node[prof] (m) [below=1.75cm of V5N] {};}

  \draw (X) -- (j);
  \draw (X) -- (k);
  \draw (XN) -- (l);
  \draw (XN) -- (m);
  
  \draw (V1) -- (V2);
  \draw (V2) -- (V3);
  \draw (V3) -- (V4);
  \draw (V4) -- (V5);
  \draw (V5) -- (V6);
  \draw (V6) -- (V7);
  \draw (V7) -- (V0);
  \draw (V0) -- (V7N);
  \draw (V7N) -- (V6N);
  \draw (V6N) -- (V5N);
  \draw (V5N) -- (V4N);
  \draw (V4N) -- (V3N);
  \draw (V3N) -- (V2N);
  \draw (V2N) -- (V1N);
  
  \draw (W0) -- (W1);
  \draw (W0) -- (W2);
  \draw (W0) -- (W3);
  \draw (W0) -- (W4);
  \draw (W0) -- (W5);
  \draw (W0) -- (W6);
  \draw (W0) -- (W7);
  \draw (W0) -- (V1);
  \draw (W0) -- (V2);
  \draw (W0) -- (V3);
  \draw (W0) -- (V4);
  \draw (W0) -- (V5);
  \draw (W0) -- (V6);
  \draw (W0) -- (V7);
  \draw (W0) -- (V0);
  \draw (W0) -- (V7N);
  \draw (W0) -- (V6N);
  \draw (W0) -- (V5N);
  \draw (W0) -- (V4N);
  \draw (W0) -- (V3N);
  \draw (W0) -- (V2N);
  \draw (W0) -- (V1N);
  
  \draw (X) -- (V1);
  \draw (X) -- (V2);
  \draw (X) -- (V3);
  \draw (X) -- (V4);
  \draw (X) -- (V5);
  \draw (X) -- (V6);
  \draw (X) -- (V7);
  \draw (XN) -- (V1N);
  \draw (XN) -- (V2N);
  \draw (XN) -- (V3N);
  \draw (XN) -- (V4N);
  \draw (XN) -- (V5N);
  \draw (XN) -- (V6N);
  \draw (XN) -- (V7N);
\end{tikzpicture}
\centering
\caption{The variable gadget.}
\label{fig:variable}
\end{figure}
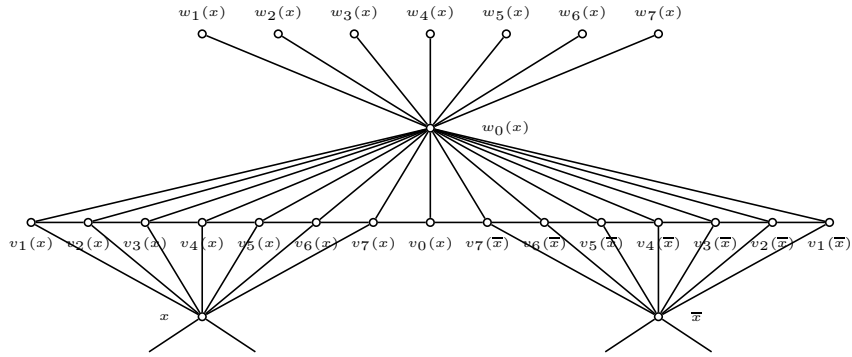

\item For each clause $c$ of $\phi$, graph $G(\phi)$ includes
a {\em clause gadget} for $c$ consisting of $18$ vertices
and $32$ edges (see Figure~\ref{fig:clause}).

The vertices of the clause gadget for $c$ are:
the {\em clause} vertex $c$;
vertices $u_1(c)$, $u_2(c)$; 
vertices $\upsilon_1(c),\ldots,\upsilon_{15}(c)$.

The $32$ edges of the clause gadget are;
edges $(c,u_1(x)), (c,u_2(x))$;
edges $(u_i(c),\upsilon_j(c))$ with $i=1, 2$ and $j=1, \ldots, 15$.

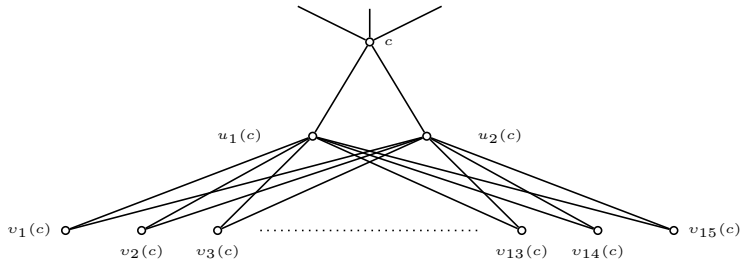
\begin{figure}[htb]
 \begin{tikzpicture}[-,auto,on grid=true,semithick,
                     prof/.style={shape=circle,draw,inner sep=0pt,minimum size=1mm},
                     every label/.style={font=\tiny}]

  \node[prof] (V1) [label={left:$\upsilon_1(c)$}] {};
  \node[prof] (V2) [right=1cm of V1] [label={below:$\upsilon_2(c)$}] {};
  \node[prof] (V3) [right=1cm of V2] [label={below:$\upsilon_3(c)$}] {};

  \phantom{\node[prof] (LL) [right=0.5cm of V3] {};}
  \phantom{\node[prof] (RR) [right=5.5cm of V1] {};}
  
  \phantom{\node[prof] (G) [right=4cm of V1] {};}
  \node[prof] (V13) [right=6cm of V1] [label={below:$\upsilon_{13}(c)$}] {};
  \node[prof] (V14) [right=1cm of V13] [label={below:$\upsilon_{14}(c)$}] {};
  \node[prof] (V15) [right=1cm of V14] [label={right:$\upsilon_{15}(c)$}] {};
  
  \phantom{\node[prof] (H) [above=1.25cm of G] {};}
  \node[prof] (U1) [left=0.75cm of H] [label={left:$u_1(c) \qquad$}] {};
  \node[prof] (U2) [right=0.75cm of H] [label={right:$\qquad u_2(c)$}] {};
  
  \node[prof] (C) [above=1.25cm of H] [label={right:$c$}] {};

  \phantom{\node[prof] (L) [above=0.5cm of C] {};}
  \phantom{\node[prof] (M) [left=1cm of L] {};}
  \phantom{\node[prof] (N) [right=1cm of L] {};}

  \draw[dotted] (LL) -- (RR);
  \draw (C) -- (L);
  \draw (C) -- (M);
  \draw (C) -- (N);
  
  \draw (C) -- (U1);
  \draw (C) -- (U2);
  \draw (U1) -- (V1);
  \draw (U1) -- (V2);
  \draw (U1) -- (V3);
  \draw (U1) -- (V13);
  \draw (U1) -- (V14);
  \draw (U1) -- (V15);
  \draw (U2) -- (V1);
  \draw (U2) -- (V2);
  \draw (U2) -- (V3);
  \draw (U2) -- (V13);
  \draw (U2) -- (V14);
  \draw (U2) -- (V15);
\end{tikzpicture}
\centering
\caption{The clause gadget.}
\label{fig:clause}
\end{figure}
 
\item For every clause $c$ in $\phi$, $G(\phi)$ contains edges between the clause vertex $c$ and the three literal 
vertices corresponding to the literals appearing $c$.
Thus, each literal vertex $x$ is connected to exactly two clause vertices, corresponding to the two clauses in which that literal appears in $\phi$.

\item Graph $G(\phi)$ contains a {\em clique} disconnected from the rest of the graph of size $N+6C+1$ vertices,
with even $N$ such that $2 \leq N \leq \frac{133 - 147\varepsilon}{4\varepsilon}C$. (Note that our choice of $\varepsilon$ implies that $\frac{133 - 147\varepsilon}{4\varepsilon}C \geq 2$.)

\item $N+\frac{123C}{4}$ additional {\em isolated} vertices.
\end{itemize}
Observe that the total number of vertices in $G(\phi)$ is $n = 2\left(N+\frac{147C}{4}\right) + 1$
(recall that $C$ is a multiple of $4$, so $\frac{147C}{4}$ is an integer).

The stubbornness level of the vertices of $G(\phi)$ is defined as follows.
All vertices belonging to gadgets and the isolated vertices have stubbornness level less than $1/2$.
Vertices in the clique are divided in two groups. $N$ clique vertices have a stubbornness level greater than $\frac{N+6C}{N+6C+1}$. We will call them \emph{asocial} vertices, 
since they never have an incentive to declare an opinion different from their own belief, whatever are the opinions declared by the remaining vertices in the clique. 
Call the remaining clique vertices $c_0, \ldots, c_{6C}$; then vertex $c_i$ has a stubbornness level $\alpha_{c_i}$ such that $\frac{N-6C+2(i-1)}{N-6C + 2(i-1)+1} < \alpha_{c_i} < \frac{N-6C+2i}{N-6C + 2i+1}$ (since $N$ is even, then the denominator is never $0$ and it always has the same sign of the numerator, from which he have that $\alpha_{c_i} \geq 0$, as desired).
This implies that vertex $c_i$ adopts an opinion different from her 
own belief if and only if at least $N+i$ clique vertices adopt that opinion.

A belief assignment to the vertices of $G(\phi)$ is called \emph{proper} if it assigns belief 1 to the following vertices:
\begin{itemize}
\item
for every variable $x$,
vertex $w_0(x)$ and only one of the two literal vertices in the gadget of $x$;

\item
for every clause $c$,
vertices $u_1(c)$ and $u_2(c)$ in the gadget of $c$;

\item 
the $N$ asocial clique vertices.

\item 
All the remaining vertices have belief 0.
\end{itemize}
Hence, in a proper profile the number of vertices with belief $1$ is
$2V + 2C + N = \frac{7C}{2} + N \leq \frac{n-1}{2}(1 - \varepsilon)$, (remind that  $N \leq \frac{133 - 147\varepsilon}{4\varepsilon}C$).

To prove Theorem \ref{thm:reduction} we will use the following two lemmas.
The first lemma proves that there exists a proper subvertable belief assignment for $G(\phi)$ if and only $\phi$ is satisfiable.
\begin{lemma}\label{lemma:proper}
The Boolean formula $\phi$ is satisfiable if and only 
if there exists a proper belief assignment to the vertices 
of $G(\phi)$ that is subvertable.
\end{lemma}
\begin{proof}
We observe that in a proper belief assignment all clique vertices eventually adopt opinion $1$, even if their belief is $0$. In fact, consider vertices $c_i$, for $i= 0, \ldots, 6C$ in this sequence: when vertex $c_i$ is selected there are $N+i$ vertices in the clique with opinion $1$ and then $c_i$ has an incentive to adopt opinion $1$.

Let $\blf$ be a proper belief assignment. Using an argument similar to what shown in \cite{acfgpWINE15}, we can prove that: 
(i) there exists a sequence of updates that leads from $\blf$ to an equilibrium in which
$17$ vertices of every variable gadget have opinion $1$
but no sequence of updates can reach an equilibrium from $\blf$ where 
more than $17$ vertices in a variable gadget have opinion $1$;
(ii) there exists a sequence of updates that leads to an equilibrium in which $17$ vertices of the clause gadget have opinion $1$;
(iii) the updates lead to an additional number of $C$ clause vertices adopting opinion $1$ in the equilibrium
if and only if $\phi$ is satisfiable.

Thus, we have that if $\phi$ is satisfiable then there exists a sequence of updates leading to an equilibrium where $N + 6C + 1 + 17V + 17C + C = N+\frac{147C}{4} + 1= \frac{n+1}{2}$ vertices have opinion $1$.
Otherwise, if $\phi$ is not satisfiable, any sequence of updates
leads to an equilibrium where less than $n/2$ vertices have opinion $1$.
\end{proof}
To conclude our proof we will show that we can ignore non-proper assignments since there is no sequence of updates that leads from a non-proper belief assignment to an equilibrium where the majority of vertices adopt opinion $1$.
\begin{lemma}\label{lemma:non-proper}
For each non-proper belief assignment $\blf$ to the vertices of $G(\phi)$ that assigns opinion $1$ to at most $\frac{7C}{2} + N$ vertices,
there is no sequence of updates that leads from $\blf$ to an equilibrium where at least $\frac{n+1}{2}$ vertices adopt opinion $1$.
\end{lemma}
\begin{proof}
We start observing that if the number of clique and isolated vertices in $G(\phi)$ adopting opinion $1$ is
strictly less than $N$, then no clique vertex with opinion $0$ will never adopt opinion $1$ (the same trivially holds for isolated vertices). 
Thus, in this case, even if it would be possible to convince all vertices in the variable and clause gadgets to adopt opinion $1$,
the number of vertices with opinion $1$ in the equilibrium would be at most $25V+18C+N-1 = N + \frac{147C}{4} - 1 < \frac{n-1}{2}$.

Let us now focus on a belief assignment that assigns belief $1$ to at most $7C/2=2C+2V$ vertices 
from variable and clause gadgets. By adopting an argument similar to what given in \cite{acfgpWINE15},
we can then prove that, if this belief assignment is such that a sequence of updates leads to an equilibrium where at least $\frac{n+1}{2}$ 
vertices adopt opinion $1$, then this belief assignment must be proper.
\end{proof}
\end{proof}

\bibliographystyle{plain}
\bibliography{opinion}

\end{document}

%% file: algo.tex
\begin{algorithm}[Hhtbp]
\DontPrintSemicolon
% \SetKwFunction{WarmUp}{WarmUp}
\SetKwFunction{MinRankInNotS}{MinRankInNotS}
\SetKwFunction{MinRankInS}{MinRankInS}
\KwIn{A graph $G$ with an odd number of vertices and at least one non-stubborn vertex}
\KwOut{A pair $(\S, u)$ where $\S$ is a good bisection and $u$ is its good vertex}
Compute a bisection $\S = (S, \notS)$ of $G$ of $3$-minimal potential \nllabel{line:bisection}\;
Let $M$ be the set of non-stubborn vertices of minimum rank in $\S$\;
\tcc{Warm-up cases}
\If{there is $u \in S$ with $\defi_\S(u) \leq -a_u - 1$}
  {Let $\T=(\notS \cup \{u\}, S \setminus \{u\})$\;
  \Return $(\T,u)$ \nllabel{line:wuc1}\;}
 \If{there is $u \in S$ with $\defi_\S(u) \geq a_u + 1$}
  {\Return $(\S,u)$ \nllabel{line:wuc2}\;}
 \If{there is $u \in \notS$ with $\defi_\S(u) \geq a_u+1$}
  {Pick $w \in S$ and let $\T=(\notS \cup \{w\}, S \setminus \{w\})$\;
  \Return $(\T,u)$ \nllabel{line:wuc3}\;}
 \If{there is $u \in S \cap M$ with $\defi_\S(u) < 0$}
  {Let $\S' = (\notS \cup \{u\}, S \setminus \{u\})$ \nllabel{line:Sp_neg_def}\;
  Pick $u$-pair $(A_u, B_u)$ for $\S'$ and let $\T$ be the associated bisection\;
  \Return $(\T,u)$  \nllabel{line:neg_def}\;}
 \tcc{There is a non-stubborn vertex of minimal rank in $\notS$}
 \If{$M \cap \notS \neq \emptyset$ \nllabel{line:second_part}}
  {\Return \MinRankInNotS{$\S$}\;}
 \tcc{All non-stubborn vertices of minimal rank are in $S$}
 \Else{\Return \MinRankInS{$\S$} \nllabel{line:third_part}\;}
 \caption{Algorithm for computing a good bisection and its good vertex}
 \label{algo}
\end{algorithm}

%% file: algo_notS.tex
\begin{algorithm}[Hhtbp]
\DontPrintSemicolon
\SetKwProg{Fn}{}{\string:}{}%
\Fn{\MinRankInNotS{$\S$}}{
  \For{every $u \in \notS \cap M$ \nllabel{line_for_notS}} 
    {\For{every $v\in S \cap N(u) \cap M$}
      {Pick a $u$-pair $(A_u, B_u)$ for $\S$ with $v\in A_u$ and let $\T$ be the associated bisection \nllabel{line:u_pair_notS_if}\;
      \If{$\T$ is good}
        {\Return $(\T, u)$ \nllabel{line:case_notS_good_if}\;}}}
  Let $y$ be an obstruction for one of the bisections $\T$ defined at Line~\ref{line:u_pair_notS_if} \nllabel{line:obstruction_notS}\;
  \If{there is non-stubborn $v \in S$}
   {Pick a $v$-pair $(A_v, B_v)$ for $\S$ with $y \in B_v$ and let $\T$ be the associated bisection \nllabel{line:pair_notS}\;
    \Return $(\T, v)$ \nllabel{line:case_notS_notgood_notstubbornS}\;}
  Pick a vertex $w \in S \cap \notN(y)$ and let $\S' = (S \setminus \{w\} \cup \{y\}, \notS \setminus \{y\} \cup \{w\})$ \nllabel{line:Sp_notS}\;
  Pick a $y$-pair $(A_y, B_y)$ for $\S'$ and let $\T$ be the associated bisection \nllabel{line:y_pair_notS}\;
  \Return $(\T, y)$ \nllabel{line:case_notS_notgood_stubbornS}\;
}
 \caption{There is a non-stubborn vertex of minimal rank in $\notS$}
 \label{algo_notS}
\end{algorithm}

%% file: algo_S.tex
\begin{algorithm}[Hhtbp]
\DontPrintSemicolon
\SetKwProg{Fn}{}{\string:}{}
\Fn{\MinRankInS{}}{
 \For{every vertex $u\in S\cap M$}
    {\For{every vertex $v\in\notS \cap N(u)$}
      {Pick a $u$-pair $(A_u, B_u)$ for $\S$ with $v\in B_u$ and let $\T$ be the associated bisection \nllabel{line:u_pair_S_if}\;
	\lIf{$\T$ is good}{\Return $(\T, u)$ \nllabel{line:case_S_good_if}}
	\If{there is an obstruction $y$ to $\T$ with $\defi_\S(y) < 0$ \nllabel{line:obstruction_S_if}}
	  {Let $\S_0 = (\notS \cup \{y\}, S \setminus \{y\})$ \nllabel{line:Sp_S_if}\;
	  Pick a $y$-pair $(A_y, B_y)$ for $\S_0$ and let $\T_0$ be the associated bisection \nllabel{line:y_pair_S_if}\;
	  \Return $(\T_0,y)$ \nllabel{line:case_S_notgood_if}\;}}}
\If{there is an obstruction $y$ to one bisection considered at Line~\ref{line:u_pair_S_if} with odd $a_y$}
  {Let $y$ be such an obstruction \nllabel{line:obstruction_y_alt}\;}
\Else{Let $y$ be an obstruction to one bisection considered at Line~\ref{line:u_pair_S_if}\nllabel{line:obstruction_y}\;}
Let $\S_1 = (\notS \cup \{y\}, S \setminus \{y\})$ \nllabel{line:S1}\;
Let $O = \left\{w \in \notS \cap \notN(y) \mid \defi_\S(w) = a_y - a_w + 1\right\}$\;
Pick a $y$-pair $(A_y, B_y)$ for $\S_1$ such that $|A_y \cap O| = \max \left\{|A_y|, |O|\right\}$ \nllabel{line:T1_pair}\;
Let $\T_1$ be the associated bisection \nllabel{line:T1}\;
 \lIf{$\T_1$ is good}{\Return $(\T_1, y)$ \nllabel{line:returnT1}}
 \If{there is $w \in \left(\notS \cup \{y\} \setminus A_y\right) \cap \notN(y)$ such that $\defi_\S(w) = a_y - a_w$ \nllabel{line:obstruction}}
  {Let $y_1$ be such a vertex \nllabel{line:obstruction_y1_alt}\;}
 \Else{Let $y_1$ be an obstruction $y_1$ to $\T_1$ \nllabel{line:obstruction_y1}\;}
 Let $\S_2 = (S_2, \notS_2) = (S \cup \{y_1\} \setminus \{y\}, \notS \cup\{y\} \setminus \{y_1\})$ \nllabel{line:S2}\;
 \If{there is a vertex $v$ with $\rank_{\S_2}(v)<\rank_{\S_2}(y)$}
  {Pick a $v$-pair $(A_v, B_v)$ for $\S_2$ and let $\T_2$ be the associated bisection \nllabel{line:T2}\;
  \lIf{$\T_2$ is good}{\Return $(\T_2, v)$ \nllabel{line:returnT2}}
  Pick an obstruction $y_2$ to $\T_2$ \nllabel{line:obstruction_y2}\;
  Let $\S_3 = (\notS_2 \cup \{y_2\}, S_2 \setminus \{y_2\})$ \nllabel{line:S3}\;
  Pick a $y_2$-pair $(A_{y_2}, B_{y_2})$ for $\S_3$ and let $\T_3$ be the associated bisection \nllabel{line:T3}\;
  \Return $(\T_3, y_2)$ \nllabel{line:returnT3}\;}
 Pick a $y$-pair $(A_{y}, B_{y})$ for $\S_2$ and let $\T_4$ be the associated bisection \nllabel{line:T4}\;
 \lIf{$\T_4$ is good}{\Return $(\T_4,y)$ \nllabel{line:returnT4}}
 \If{there is $w \in S \cup \{y_1\} \setminus \{y\}$ of minimum rank for $\S_2$ with $\defi_{\S_2}(w) < 0$}
  {Let $\S_4 = (\notS_2 \cup \{w\}, S_2 \setminus \{w\})$ \nllabel{line:Sp_neg_defS}\;
  Pick $w$-pair $(A_w, B_w)$ for $\S_4$ with $y \in A_w$ and let $\T_5$ be the associated bisection \nllabel{line:pair_neg_defS}\;
  \Return $(\T_5,w)$  \nllabel{line:neg_defS}\;}
 Pick an obstruction $y_4$ to $\T_4$ \nllabel{line:obstruction_y4}\;
 Pick a $y_1$-pair $(A_{y_1}, B_{y_1})$ for $\S_2$ with $y_4 \in B_{y_1}$ and let $\T_6$ be the associated bisection \nllabel{line:T5}\;
 \lIf{$\T_6$ is good}{\Return $(\T_6, y_1)$ \nllabel{line:returnT5}}
 Pick an obstruction $y_6$ to $\T_6$ \nllabel{line:obstruction_y6}\;
 
 \If{$W(y_6, y) = W(y_4, y) = W(y_6, y_4) = 1$ and $W(y_1, y_6)$}
 {Let $\S_5 = (\notS_2 \cup \{y_6\}, S_2 \setminus \{y_6\})$ \nllabel{line:S7}\;
 Pick a $y$-pair $(A_y, B_y)$ for $\S_5$ and let $\T_7$ be the associated bisection \nllabel{line:T9}\;
 \Return $(\T_7, y)$ \nllabel{line:returnT9}\;}
 
 Let $\S_6 = (S_6, \notS_6) = (S_2 \cup \{y_4\} \setminus \{y_6\}, \notS_2 \cup\{y_6\} \setminus \{y_4\})$ \nllabel{line:S5}\;
 
 \If{there is a non-stubborn vertex $x$ in $\notS_6$ with $\rank_{\S_6}(x) = \ell - 1$}
   {\Return \MinRankInNotS{$\S_6$} \nllabel{line:lastreturn}\;}
 
 Pick a $y_1$-pair $(A_{y_1}, B_{y_1})$ for $\S_6$ and let $\T_8$ be the associated bisection \nllabel{line:T7}\;
 \lIf{$\T_8$ is good}{\Return $(\T_8, y_1)$ \nllabel{line:returnT7}}
 Pick an obstruction $y_8$ to $\T_8$ \nllabel{line:obstruction_y7} and let $\S_7 = (\notS_6 \cup \{y_8\}, S_6 \setminus \{y_8\})$ \nllabel{line:S6}\;
 Pick a $y_8$-pair $(A_{y_8}, B_{y_8})$ for $\S_7$ and let $\T_9$ be the associated bisection \nllabel{line:T8}\;
 \Return $(\T_9, y_8)$ \nllabel{line:returnT8}\;
}
 \caption{All non-stubborn vertices of minimal rank are in $S$}
 \label{algo_S}
\end{algorithm}